\documentclass[12pt, draftclsnofoot, onecolumn]{IEEEtran}

\usepackage{amsmath, amssymb, comment}
\usepackage{graphicx, color, tabu}
\usepackage[latin1]{inputenc}
\usepackage{caption, subcaption, tikz}
\usepackage{relsize, dsfont, mathtools}

\usetikzlibrary{patterns}

\newtheorem{theorem}{Theorem}
\newtheorem{corollary}{Corollary}

\newtheorem{proposition}{Proposition}
\newtheorem{proof}{Proof}

\newcommand{\sinr}{\mathsf{SINR}}
\newcommand{\PP}{\mathbb{P}}
\newcommand\E{\mathbb{E}}
\newcommand{\La}{\mathcal{L}}
\newcommand{\dd}{{\rm d}}
\newcommand{\R}{\mathbb{R}}
\newcommand{\G}{\Gamma}
\newcommand{\g}{\gamma}
\newcommand{\la}{\lambda}

\newcommand{\TWD}{\mathsf{2D}}
\newcommand{\TWU}{\mathsf{2U}}
\newcommand{\THD}{\mathsf{3D}}
\newcommand{\THU}{\mathsf{3U}}
\newcommand{\TWN}{\mathsf{2N}}
\newcommand{\THN}{\mathsf{3N}}
\newcommand{\supfn}{f_\ell(\theta, \theta_{\rm max})}
\newcommand{\Mod}[1]{\ \left(\text{mod}\ #1\right)}
\newcommand{\mfrac}[2]{\genfrac{}{}{0pt}{}{#1}{#2}}
\DeclareMathOperator*{\argmax}{arg\,max}

\IEEEoverridecommandlockouts
\allowdisplaybreaks
\graphicspath{{figures/}}

\begin{document}

\title{Impact of Directionality on Interference Mitigation in Full-Duplex Cellular Networks}

\author{Constantinos Psomas, \IEEEmembership{Member, IEEE}, Mohammadali Mohammadi, \IEEEmembership{Member, IEEE}, Ioannis Krikidis, \IEEEmembership{Senior Member, IEEE}, and Himal A. Suraweera, \IEEEmembership{Senior Member, IEEE}\vspace{-1cm}

\thanks{C. Psomas and I. Krikidis are with the  KIOS Research Center for Intelligent Systems and Networks, University of Cyprus, Cyprus (e-mail: \{psomas, krikidis\}@ucy.ac.cy).

M. Mohammadi is with the Faculty of Engineering, Shahrekord University, Shahrekord 115, Iran (email: m.a.mohammadi@eng.sku.ac.ir).

H. A. Suraweera is with the Department of Electrical and Electronic Engineering, University of Peradeniya, Peradeniya 20400, Sri Lanka (email: himal@ee.pdn.ac.lk).

This work was supported by the Research Promotion Foundation, Cyprus under the project FUPLEX with pr. no. CY-IL/0114/02. Part of this work was presented at the IEEE International Workshop on Signal Processing Advances in Wireless Communications, Stockholm, Sweden, 2015 \cite{SPAWC}.}}

\maketitle
\begin{abstract}
In this paper, we consider two fundamental full-duplex (FD) architectures, two-node and three-node, in the context of cellular networks where the terminals employ directional antennas. The simultaneous transmission and reception of data in non-orthogonal channels makes FD radio a potential solution for the currently limited spectrum. However, its implementation generates high levels of interference either in the form of loopback interference (LI) from the output to the input antenna of a transceiver or in the form of co-channel interference in large-scale multicell networks due to the large number of active links. Using a stochastic geometry model, we investigate how directional antennas can control and mitigate the co-channel interference. Furthermore, we provide a model which characterizes the way directional antennas manage the LI in order to passively suppress it. Our results show that both architectures can benefit significantly by the employment of directional antennas. Finally, we consider the case where both architectures are employed in the network and derive the optimal values for the density fraction of each architecture which maximize the success probability and the network throughput.
\end{abstract}


\section{Introduction}\label{sec:intro}
\IEEEPARstart{I}{nterference} is a fundamental notion in wireless communications. Its existence is an inevitable outcome of the concurrent use of wireless resources between multiple transmitters, that is frequency, code or time. Conventionally, the concept of orthogonality is applied to reduce it or maybe even eliminate it. For instance, in cellular networks, terminals in the same cell transmit using different carrier frequency or time slot thus restricting the co-channel interference at a receiver to out-of-cell transmitters. Furthermore, the recent Long Term Evolution (LTE) standard, implements Orthogonal Frequency Division Multiplexing (OFDM), which divides the available bandwidth into a large set of sub-carriers which are transmitted in parallel. The division is done in a way such that the frequency space between the sub-carriers is minimized but orthogonality is still achieved. However, even though orthogonality assists in the reduction of co-channel interference, it limits the available spectrum. Towards this direction, full-duplex (FD) is considered a potential technology for the next generation of communication networks \cite{JSAC, HA}.

FD is a well investigated technology which could potentially double the available spectrum and subsequently increase the data rate compared to half-duplex (HD) radio, as it employs simultaneous transmission and reception using non-orthogonal channels \cite{JSAC}. Despite its promising potential, FD radio has been overlooked, especially for large-scale multicell networks due to the high levels of interference it generates. The use of non-orthogonal channels has the critical disadvantage of increasing the interference in a cellular network, which significantly degrades its performance \cite{CP}. The existence of more active wireless links results in the escalation of both intra- and out-of-cell co-channel interference. Moreover, the non-orthogonal operation at a transceiver creates a loopback interference (LI) between the input and output antennas \cite{TR1}-\cite{BHA}. This aggregate interference at a receiver is why FD has been previously regarded as an unrealistic approach in wireless communications. In particular, the main reason is the LI which was considered to make wireless communications impractical. Consequently, the primary concern towards making FD feasible, was how to mitigate the LI and with the advancements in signal processing and antennas, many methods now exist to achieve this \cite{JSAC, HA}. These methods can be passive (channel-unaware), e.g. \cite{AS1}-\cite{EE2}, active (channel aware), e.g. \cite{TR2, NGO}, or a combination of the two.

The existence of interference in wireless networks has urged researchers to consider methods to either exploit it is such a way as to achieve power savings \cite{MZ} or manage it in a manner that would achieve performance gains. In this paper we focus on the latter case and consider a well-known method which is directional transmission and reception \cite{DIR1}, \cite{DIR2}. In the omni-directional case, the signal is transmitted in all directions and, as a result, interferes with all other terminals in the network. Therefore, by focusing the signal to a certain direction reduces the number of terminals that are affected by the interfering signal, i.e. the terminals that lie in the transmitted direction. Furthermore, compared to the omni-directional case, the directed transmitted signal can achieve a longer distance with the same power and can also reach the receiver with higher power at the same distance. As the beamwidth decreases, the gain of the signal increases and the possibility of interfering with other terminals decreases. The significance of directional antennas in large-scale multicell networks has been shown before in various contexts. In \cite{AH}, the authors studied an ad-hoc network's performance under some spatial diversity methods and showed the achieved gains. The work in \cite{JW} developed a model to investigate the impact of beam misdirection on the network's performance. The impact on the performance of a receiver in a heterogeneous HD cellular network with directional antennas is demonstrated in \cite{Wang}. Finally, the authors of \cite{TB}, provide a performance analysis of mmWave cellular networks with blockage where directional antennas are essential.

Apart from the reduction in co-channel interference, the employment of directional antennas in an FD context provides the prospect of passively suppressing the LI with antenna separation techniques \cite{AS1}, \cite{EE1}. The angle formed between the transmit and receive antennas when they point to different directions reduces the intensity of the LI and thus the final residual LI after active cancellation is minimized. Indeed, in order to bring down the residual LI as close to the noise floor as possible, in addition to active cancellation, passive suppression is also required \cite{AS3}. Given these observations, the use of directional antennas in large-scale multicell FD networks seems as a promising solution to manage and control the high levels of interference. Previous studies on large-scale FD networks were mostly concerned with single antenna scenarios. In \cite{TONG}, the authors considered a wireless ad-hoc network with both FD and HD capabilities and showed that under imperfect LI cancellation there exists a break-even point where FD and HD provide the same performance. Cellular networks were considered in \cite{GOY} where FD was implemented only at the base stations (BSs); the authors showed that the uplink is more affected by the interference compared to the downlink. The effect of multi-user interference in FD cellular networks was studied in \cite{RKM}; it was shown that without dedicated interference management, FD is not feasible in macrocell networks but it can be viable in microcell networks under certain conditions. Moreover, the authors of \cite{GOY2} studied a hybrid HD/FD cellular network and demonstrated the trade-off between the average spectral efficiency and coverage with respect to the number of FD BSs. The work in \cite{CP} considered both FD and HD-enabled users in cellular networks, showing that FD can increase the downlink performance if the LI can be significantly mitigated. In \cite{LEE}, hybrid HD/FD multi-tier cellular networks were investigated and it was shown that in order to maximize the network's throughput different tiers should operate in different duplex modes. Hybrid HD/FD cellular networks were also studied in \cite{AL} where the authors considered both cell center and cell edge users together with realistic parameters such as pulse shaping and matched filtering; it was demonstrated that FD BSs with HD users provide higher performance than FD BSs and FD users. A single-cell scenario was investigated in \cite{MM} with a multiple antenna FD access point and single antenna HD users; it was demonstrated that the average sum rate can be increased with the employment of linear precoding. Finally, the work in \cite{KOUNT} looked at FD small-cell multiple-input multiple-output (MIMO) relays and derived tight bounds for the success probability.

\subsection{Paper contributions}
In this paper, we study the performance of two FD architectures, two-node and three-node \cite{JSAC}, in cellular networks where the terminals employ directional antennas to manage and thus mitigate the high-levels of interference in the network. Specifically, the paper's contributions are as follows
\begin{itemize}
  \item We derive analytical expressions for the outage probability of the network for each architecture using stochastic geometry \cite{HAEN1}, and show that with the employment of directional antennas, the co-channel interference can be regulated in such a way as to significantly reduce it at a terminal and thus improve its performance.
  \item We derive a simple mathematical model which characterizes the behaviour of directional antennas with regards to the mitigation of the LI. Our model provides the level of the passive LI suppression at a BS as a function of the angle between the transmit and receive antennas. Our model is based on the experimental results in \cite{EE1} but we have generalized it for any practical scenario.
  \item We study the asymptotic cases when the number of employed antennas and the density of the network become large. We show that the performance of the three-node architecture is improved with the employment of more antennas. On the other hand, the performance of the two-node architecture degrades for a large number of antennas due to the high LI power gain as both transmit and receive antenna operate in the same direction. Furthermore, we show that denser networks improve the performance of both architectures and for ultra-dense networks the performance is independent of the LI.
  \item Finally, we consider the composite architecture case where both architectures are employed in the network and provide analytical expressions for the success probability and network throughput. We derive the optimal values for the density of each architecture in the composite case and show that the three-node architecture is preferred in most scenarios.
\end{itemize}
Our results show the significant gains that can be achieved by the employment of directional antennas and also show that the three-node architecture performs better due to the passive suppression of the LI.
 
The rest of the paper is organized as follows: Section \ref{sec:model} presents the network model together with the channel, interference and sectorized directional antenna model. Section \ref{sec:analysis} provides the main results for the outage probability of both downlink and uplink for both FD architectures together with special cases with closed-form expressions. In Section \ref{sec:het} we consider the composite case where both architectures are employed in the network and in Section \ref{sec:num} the simulation results are provided. Finally, the conclusion of the paper is given in Section \ref{sec:conc}.

\textit{Notation}: $\R^d$ denotes the $d$-dimensional Euclidean space, $\|x\|$ denotes the Euclidean norm of $x \in \R^d$, $\PP(X)$ denotes the probability of the event $X$ and $\E(X)$ represents the expected value of $X$, $\mathds{1}_{X}$ is the indicator function of $X$ with $\mathds{1}_{X} = 1$ if $X$ is true and $\mathds{1}_{X} = 0$ otherwise, $\csc(\theta)$ is the cosecant of angle $\theta$ and $G_{p q}^{m n} \left( x ~ \Big | ~ \mfrac{a_1,\dots, a_p}{b_1, \dots, b_q} \right)$ denotes the Meijer G-function \cite[Eq. (9.301)]{GRA}. Furthermore, $_2F_{1}(a,b;c;x)$ is the Gauss hypergeometric function \cite[Eq. (9.100)]{GRA} and we define $F(x,y) \triangleq {}_2F_{1}\left(1, 1-\frac{2}{x}; 2-\frac{2}{x}; -y\right)$.

\begin{figure}[t]\centering
\begin{subfigure}{0.35\linewidth}\centering
  \includegraphics[width=0.7\linewidth]{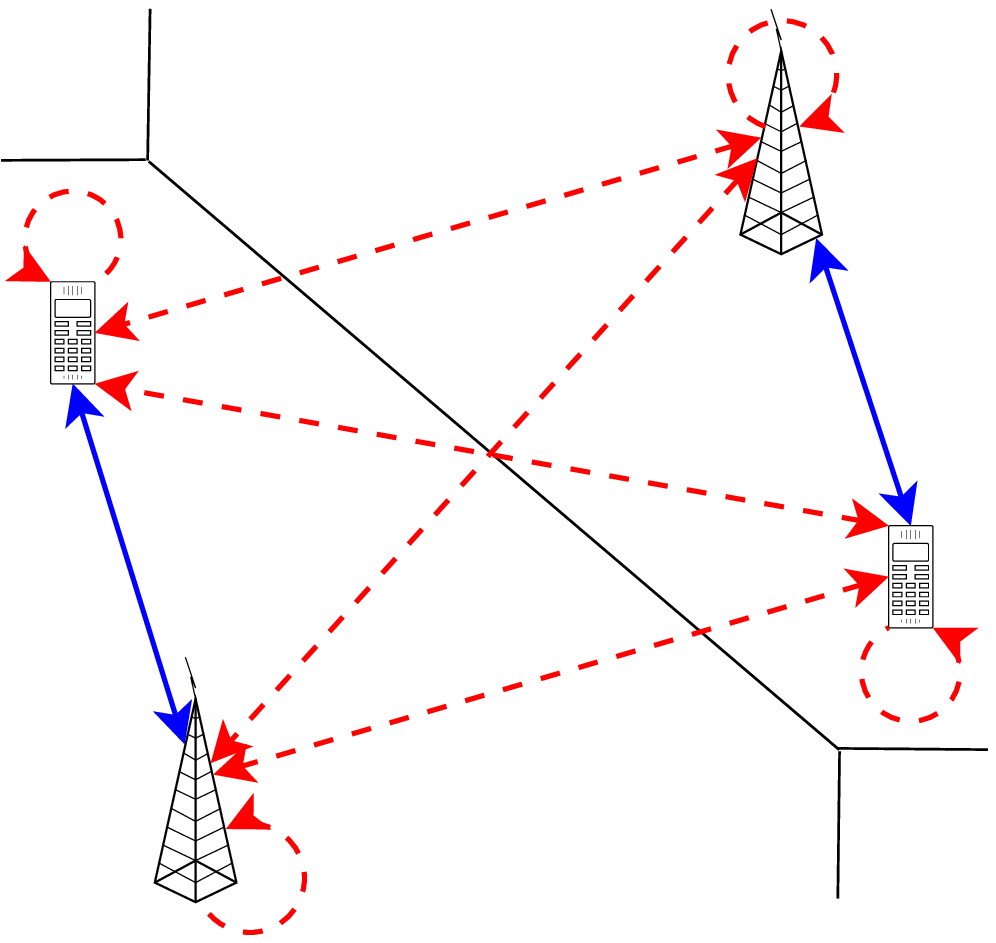}\vspace{-2mm}
  \caption{Two-node.}
  \label{fig:sce1}
\end{subfigure}\hspace{3cm}
\begin{subfigure}{0.35\linewidth}\centering
  \includegraphics[width=0.7\linewidth]{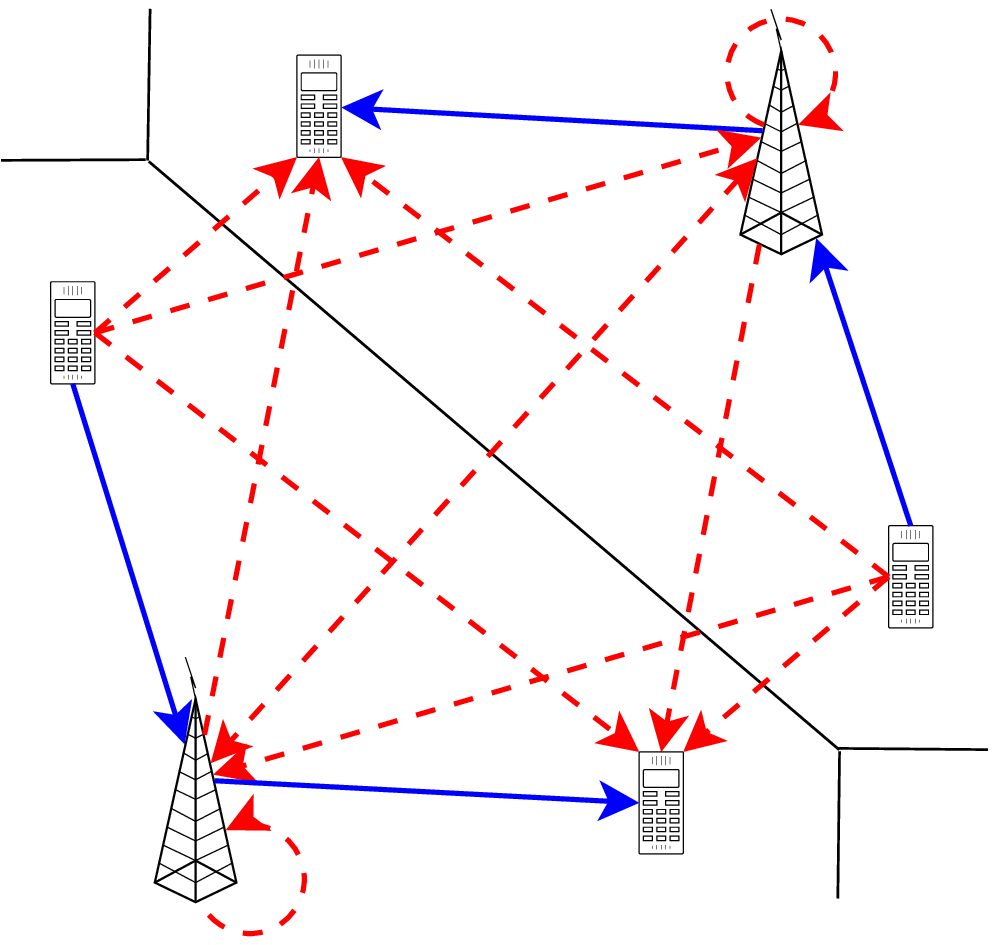}\vspace{-2mm}
  \caption{Three-node.}
  \label{fig:sce2}
\end{subfigure}
\caption{Two-node and three-node full-duplex architectures.}\label{fig:scenarios}\vspace{-10mm}
\end{figure}

\section{System Model}\label{sec:model}
FD networks can be categorized into two-node and three-node architectures \cite{JSAC}. The former, referred also as bidirectional, describes the case where both nodes, i.e., the user and the BS, have FD-capabilities. The latter describes the case where only the BS has FD-capabilities and the users operate in HD-mode. In what follows, we consider both architectures in the case where each node employs a number of directional antennas.

\subsection{Network model}
The network is studied from a large-scale point of view using stochastic geometry \cite{HAEN1}. The locations of the BSs follow a homogeneous Poisson point process (PPP) $\Phi = \{ x_i: i = 1,2,\dots\}$ of density $\la$ in the Euclidean plane $\R^2$, where $x_i \in \R^2$ denotes the location of the $i^{\rm th}$ BS. Similarly, let $\Phi' = \{ y_i: i = 1,2,\dots \}$ be a homogeneous PPP of density $\la' \gg \la$, independent of $\Phi$, representing the locations of the uplink users. A user selects to connect to the nearest BS in the plane, that is, user $i$ is served by BS $j$ if and only if $\|y_i-x_j\| < \|y_i-x_k\|$ where $x_k \in \Phi$ and $k \neq j$. Due to this inequality, the Voronoi cell formed by a BS contains multiple users ($\la'/\la$ on average) and the BS serves each uplink user in it's cell on a different channel. As such, the PPP $\Phi'$ of the uplink users in the network operating on the same channel is a thinned point process $\Psi$ with density $\lambda$. Obviously, these assumptions form correlations between the locations of the BSs and of the users so, in order to achieve tractability, we assume that the point process $\Psi$ is an independent PPP \cite{AL, JA}. Assuming the user is located at the origin $o$ and at a distance $r$ to the nearest BS, the probability density function (pdf) of $r$ is $f_r(r) = 2\pi \la re^{-\la \pi r^2}, ~r \geq 0$ \cite{HAEN1}; this distribution is also valid for the distance between two users and between two BSs. Finally, assume that all BSs transmit with the same power $P_b$ and all users with the same power $P_u$.

\subsection{Channel model}
All channels in the network are assumed to be subject to both small-scale fading and large-scale path loss. Specifically, the fading between two nodes is Rayleigh distributed and so the power of the channel fading is an exponential random variable with unit mean. The channel fadings are considered to be independent between them. The unbounded single-slope path loss model $\ell(x,y) = \|x-y\|^{-\alpha}$ is used which assumes that the received power decays with the distance between the transmitter $x$ and the receiver $y$, where $\alpha > 2$ denotes the path loss exponent; multi-slope path loss models \cite{PL1, PL2} are left for future consideration. Note that even though the bounded path loss model is more practical, we use the unbounded one to simplify our analysis. Furthermore, the effect of both models on the signal-to-interference-plus-noise ratio (SINR) statistics is small \cite{HAEN2}. Throughout this paper, we will denote the path loss exponent for the channels between a BS and a user by $\alpha_1$. The path loss exponent associated with the interfering signal propagation between two users and between two BSs will be denoted by $\alpha_2$. In reality, the path loss exponents for the signals between two BSs and between two users is different but we make this simplification since the interferences between users and between BSs are considered in independent scenarios, Section \ref{subsec:out_down} and Section \ref{subsec:out_up} respectively, and so it does not affect our analytical results and avoids notational overhead. Lastly, we assume all wireless links exhibit additive white Gaussian noise (AWGN) with zero mean and variance $\sigma_n^2$.

\subsection{Sectorized directional antennas}
Define as $M_b$ and $M_u$ the number of directional transmit/receive antennas employed at a BS and a user respectively. The main and side lobes of each antenna are approximated by a circular sector as in \cite{AH}. Therefore, the beamwidth of the main lobe is $2\pi/M_i$, $i \in \{b, u\}$. We assume that the antenna gain of the main lobe is $G_i = \frac{M_i}{1+\g_i(M_i-1)}$ where $\g_i$, $i \in \{b, u\}$ is the ratio of the side lobe level to the main lobe level \cite{AH}. Therefore, the antenna gain of the side lobe is $H_i = \g_i G_i$, $i \in \{b, u\}$. The antenna gain refers to the ability of the directional antenna to focus its energy to the intended direction and the gain is referenced to an omni-directional antenna; $M_b = M_u = 1$ refers to the omni-directional case \cite{CP}. It is assumed that the BSs employ highly adaptive directional antennas and so an active link between a user and a BS lies in the boresight direction of the antennas of both nodes \cite{TSR}, i.e., maximum power gain can be achieved\footnote{Each terminal is equipped with a set of phase shifts at the antenna elements which provide the appropriate beam pattern \cite{ANT}.}.

\begin{table}[t]\centering\tabulinesep=1mm
  \begin{tabu}{| c | c | c | c | c |}\hline
    $k$ & 1 & 2 & 3 & 4\\ \hline
    $\la_{i,j,k}$ & $\frac{\la}{M_i M_j}$ & $\frac{\la(M_j-1)}{M_i M_j}$ & $\frac{\la(M_i-1)}{M_i M_j}$ & $\frac{\la(M_i-1)(M_j-1)}{M_i M_j}$\\ \hline
    $\G_{i,j,k}$ & $G_iG_j$ & $G_iH_j$ & $G_jH_i$ & $H_iH_j$\\ \hline
  \end{tabu}
\caption{Densities $\la_{i,j,k}$ and power gains $\G_{i,j,k}$ for each thinning process $k \in \{1,2,3,4\}$, $i,j \in \{b, u\}$.}\label{tbl:thin}\vspace{-9.5mm}
\end{table}

\subsection{Interference}\label{subsec:inter}
The total co-channel interference at a node is the aggregate sum of the interfering received signals from the BSs of $\Phi$ and the uplink users of $\Psi$; we assume that the received interfering signals at a node are uncorrelated. In the two-node architecture, co-channel interference at any node results from both out-of-cell uplink users and BSs. In the three-node architecture, the BS experiences co-channel interference from out-of-cell BSs and uplink users, whereas the receiver experiences additional intra-cell interference from the uplink user. When $M_b > 1$ or $M_u > 1$ the transmitters can interfere with a receiver in four different ways \cite{AH}:
\begin{itemize}
\item[1.] Transmitting towards a receiver in the main sector,\vspace{-1mm}
\item[2.] Transmitting away from a receiver in the main sector,\vspace{-1mm}
\item[3.] Transmitting towards a receiver outside the main sector,\vspace{-1mm}
\item[4.] Transmitting away from a receiver outside the main sector,\vspace{-1mm}
\end{itemize}
where the main sector is the area covered by the main lobe of the receiver. Consider the interference received at a node $x^i \in \Phi \cup \Psi$ from all other network nodes $y^j \in \Phi \cup \Psi$, $i,j \in \{b,u\}$, $x^i \neq y^j$. To evaluate the interference, each case $k \in \{1,2,3,4\}$ needs to be considered separately. This results in each of the PPPs $\Phi$ and $\Psi$ being split into four thinning processes $\Phi_k$ and $\Psi_k$ with densities $\la_{i,j,k}$. Additionally, the power gain $\G_{i,j,k}$ of the link between $x^i$ and $y^j$ is defined as the product of the gains of the antennas associated with the link. Table \ref{tbl:thin} provides the density and power gain for each case. Note that $\sum_{k=1}^4 \la_{i,j,k} = \la$ and when $M_b = M_u = 1$ the links have no gain, i.e., $\G_{i,j,k}=1 ~\forall~ i,j,k$.

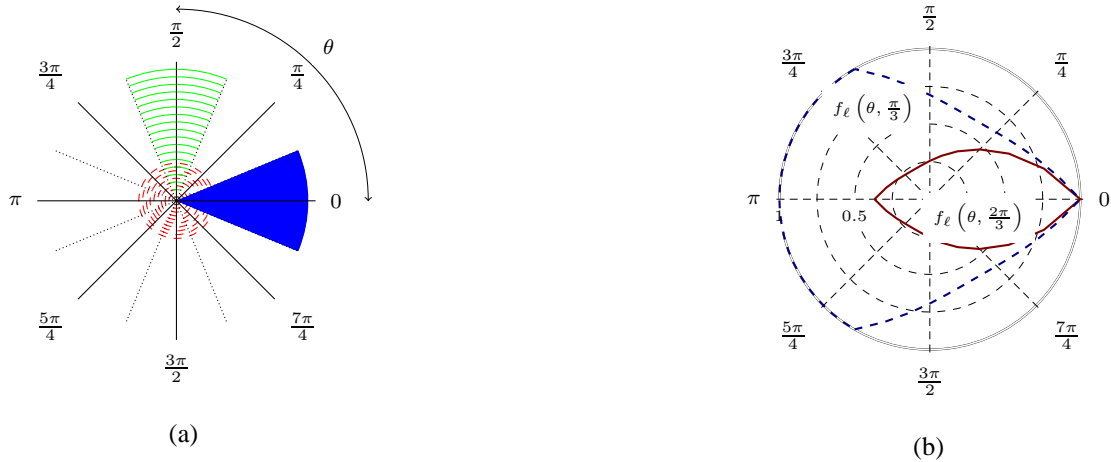
\begin{figure}[t]
  \begin{subfigure}{0.4\linewidth}\centering
    \begin{tikzpicture}[scale=0.5]
    \foreach \x/\text in {0.1,0.2,...,1} \draw [dashed,thin,red] circle(\x);
    \foreach \x/\text in {0.01,0.02,...,3.5} {
      \draw [-,thin,blue] (-22.5:\x) arc(-22.5:22.5:\x);
    }
    \foreach \x/\text in {0.1,0.3,...,3.5} {
      \draw [-,thin,green] (67.5:\x) arc(67.5:112.5:\x);
    }
    
    \foreach \ang/\lab/\dir in {
      0/{0}/right,
      1/{\frac{\pi}{4}}/above right,
      2/{\frac{\pi}{2}}/above,
      3/{\frac{3\pi}{4}}/above left,
      4/{\pi}/left,
      5/{\frac{5\pi}{4}}/below left,
      6/{\frac{3\pi}{2}}/below,
      7/{\frac{7\pi}{4}}/below right} {
      \draw (0,0) -- (\ang*45:3.7); \draw[densely dotted] (0,0) -- (22.5+\ang*45:3.5);
      \node [fill=white] at (\ang*45:3.8) [\dir] {\scriptsize $\lab$};
    }
    
    \draw [<->,ultra thin] (0:5.1) arc(0:90:5.1);
    \node [fill=white] at (45:5.1) [above right] {\scriptsize $\theta$};
    \end{tikzpicture}
    \caption{}\label{fig:ant_angle}
  \end{subfigure}\hfill
  \begin{subfigure}{0.4\linewidth}\centering
    \begin{tikzpicture}[scale=0.5]
    
    \draw[dashed, ultra thin] (0,0) circle (1);
    \draw[dashed, ultra thin] (0,0) circle (2); \node [fill=white] at (-2, 0) [below] {\tiny $0.5$};
    \draw[dashed, ultra thin] (0,0) circle (3);
    \draw[style=double, ultra thin] (0,0) circle (4);
    \node[fill=white] at (-4, 0) [below] {\tiny $1$};
    
    \foreach \ang/\lab/\dir in {
      0/0/right,
      1/{\frac{\pi}{4}}/{above right},
      2/{\frac{\pi}{2}}/above,
      3/{\frac{3\pi}{4}}/{above left},
      4/{\pi}/left,
      5/{\frac{5\pi}{4}}/{below left},
      7/{\frac{7\pi}{4}}/{below right},
      6/{\frac{3\pi}{2}}/below} {
      \draw[densely dashed, ultra thin] (0,0) -- (\ang * 180 / 4:4.2);
      \node [fill=white] at (\ang * 180 / 4:4.2) [\dir] {\scriptsize $\lab$};
    }
    
    \draw[red!50!black, thick] plot [domain=-180:180]
    (xy polar cs:angle=\x, radius={4*exp(cos(2*180/3)-cos(abs(\x)-2*180/3))});
    \draw[blue!50!black, dashed, thick] plot [domain=-180+180/3:180-180/3]
    (xy polar cs:angle=\x, radius={4*exp(cos(180/3)-cos(abs(\x)-180/3))});
    \draw[blue!50!black, dashed, thick] plot [domain=180-180/3:180+180/3]
    (xy polar cs:angle=\x, radius=4);
    \node [fill=white] at (-1.5,2.4) {\tiny $f_\ell\left(\theta, \frac{\pi}{3}\right)$};
    \node [fill=white] at (1.3,-0.5) {\tiny $f_\ell\left(\theta, \frac{2\pi}{3}\right)$};
    \end{tikzpicture}
    \caption{}\label{fig:pass_cancel}
  \end{subfigure}\vspace{-3mm}
  \caption{(a) Angle $\theta$ between antennas for $M_b = 8$. Dots correspond to the boundaries of each sector. The shaded area, solid lines and dashed lines depict the main transmission lobe, main reception lobe and side lobes respectively. (b) LI passive suppression efficiency with respect to the angle $\theta$ where the solid line depicts $\theta_{\rm max} = \frac{2\pi}{3}$ and the dashed line $\theta_{\rm max} = \frac{\pi}{3}$. The value of $\supfn$ corresponds to the fraction of the LI which cannot be suppressed.}\vspace{-8mm}
\end{figure}

Regarding the LI, we consider a model that captures the effects of both active cancellation and passive suppression. Emphasis is given to the passive suppression which is critical in mitigating the LI to the noise floor \cite{AS3}. We assume that FD-capable users and BSs employ imperfect active cancellation mechanisms. As such, we consider the residual LI channel coefficient to follow a complex Gaussian distribution with zero mean and variance $\sigma_\ell^2$ since each implementation of the cancellation mechanism can be characterized by a specific residual power \cite{TR3}, \cite{HAS}. We assume that the users employ the same imperfect active LI cancellation mechanisms and so the LI channel coefficients at the users have the same variance $\sigma_\ell^2$. Likewise, we assume that the BSs employ the same, but different to the users, imperfect active cancellation mechanisms. Furthermore, the BSs in the three-node topology are assumed to have the ability to passively suppress the LI with antenna separation techniques \cite{EE1, EE2}. This is possible since transmission and reception may be operated at different directions, separating the antennas by a certain angle (see Fig. \ref{fig:sce2}). Observe that passive suppression in the two-node architecture is not possible as both antennas always point to the same direction (see Fig. \ref{fig:sce1}). We model the effect of the passive suppression in the following way. Let $\theta \in [-\pi, \pi)$ be the angle between the two antennas (Fig. \ref{fig:ant_angle}). Let $\supfn$ denote the fraction of the LI that cannot be passively suppressed at an angle $\theta$, e.g., $\supfn = 1$ means zero passive suppression, and we assume it is given by,
\begin{align}\label{eq:ang_fun}
  \supfn = \min\left\{1,\exp\left(\cos\left(\theta_{\rm max}\right)-\cos\left(|\theta|-\theta_{\rm max}\right)\right)\right\},
\end{align}
where $\theta \in [-\pi, \pi), \theta \equiv 0 \Mod{\frac{2\pi}{M_b}}$ and $\theta_{\rm max} \in \left(0,\pi\right]$ is the angle where the maximum suppression is achieved; Appendix \ref{pass_suppress_desc} provides a detailed description of the above function. Note that the level of achievable passive LI suppression, and consequently the value of $\theta_{\rm max}$, depends on various factors such as the efficiency of antenna directionality and the environment (i.e. reflective or non-reflective) \cite{EE2}, \cite{AS3}; we assume that $\theta_{\rm max}$ increases with the antenna efficiency. Fig. \ref{fig:pass_cancel} depicts the level of passive suppression with respect to the angle $\theta$ for $\theta_{\rm max} = \frac{2\pi}{3}$ \cite{EE1} and $\theta_{\rm max} = \frac{\pi}{3}$.

\section{Performance Analysis}\label{sec:analysis}
In this section, we derive analytically the outage probability of a cellular network at both the downlink and uplink for both architectures outlined above. For the sake of fairness, user association and high-layer signaling is not taken into account and the performance of both architectures is evaluated at the physical layer in terms of the outage probability. The outage probability describes the probability that the instantaneous achievable rate of the channel is less than a fixed target rate $R$, i.e. $\PP[\log_2(1+\sinr) < R]$. Without loss of generality and following Slivnyak's Theorem \cite{HAEN1}, we execute the analysis for a typical node located at the origin but the results hold for all nodes of the network. Note that, throughout the rest of the paper, we will use the term ``receiver'' to refer to the downlink user. We denote by $u_o$ the typical receiver and by $b_o$ the typical BS and assume $b_o$ is the nearest BS to $u_o$ at a random distance $r$. Similar notation will be used for the typical nodes in the analysis of both downlink and uplink with the node of interest in each case being located at the origin.

The typical receiver $u_o$ experiences co-channel interference from the uplink users and the BSs in the network. Let $I_u$ and $I_b$ be the aggregate interference received at $u_o$ from the uplink users and the BSs (apart from $b_o$) respectively. Then $I_u$ and $I_b$ can be expressed as,
\begin{align}
I_b & = P_b \sum_{i \in \{1,2,3,4\}} \G_{u,b,i} ~ \sum_{j\in\Phi_i \setminus b_o} |g_j|^2 d_j^{-\alpha_1},\label{eq:intb}\\
I_u & = P_u \sum_{i \in \{1,2,3,4\}} \G_{u,u,i} ~~ \sum_{j\in\Psi_i} ~ |k_j|^2 D_j^{-\alpha_2},\label{eq:intu}
\end{align}
where $|g_i|^2$ and $|k_j|^2$ are the channel gains between $u_o$ and the $i^{\rm th}$ BS and $u_o$ and the $j^{\rm th}$ uplink user respectively; similarly, $d_i$ and $D_j$ are the distances between $u_o$ and the $i^{\rm th}$ BS and $u_o$ and the $j^{\rm th}$ uplink user respectively. Then, the SINR at the typical receiver $u_o$ can be written as,
\begin{align}\label{eq:sinr}
\mathsf{SINR} = \frac{P_b \G_{u,b,1} |h|^2 r^{-\alpha_1}}{\sigma_n^2 + \mathds{1}_{\rm FD}I_\ell + I_b + I_u},
\end{align}
where $|h|^2$ is the channel gain between $u_o$ and $b_o$; $\mathds{1}_{\rm FD}$ is the indicator function for the event ``$u_o$ is FD-capable"; $I_\ell$ is the residual interference at the typical node after LI cancellation and is defined as $I_\ell = P_u \G_{u,u,1} |h_\ell|^2$, where $|h_\ell|^2 \sim \exp\left(1/\sigma^2_\ell\right)$ is the residual LI channel gain at $u_0$.

The co-channel interference experienced at the typical BS $b_o$ as well as the SINR at $b_o$ can be derived in an analogous manner to above and therefore we omit their inclusion. Throughout the paper, we will use $\TWN, \TWD$ and $\TWU$ as subscripts or superscripts accordingly to refer to the two-node architecture, the two-node downlink case and the two-node uplink case respectively. Similarly, we will use $\THN, \THD$ and $\THU$ for the three-node scenario.

\subsection{Outage probability at the downlink}\label{subsec:out_down}
In what follows, we present the theorems that characterize the outage probability of an FD cellular network in the case where the two-node architecture is employed (Theorem \ref{thm:out_prob_2d}) and also in the case where the three-node architecture is employed (Theorem \ref{thm:out_prob_3d}).
\begin{theorem}\label{thm:out_prob_2d}
The outage probability of a typical receiver in the two-node architecture is
\begin{align}
P_{\TWD} = 1 - 2 \pi \la \int_0^\infty r \exp(-\la \pi r^2-s\sigma_n^2)~
\La^{\TWD}_{I_\ell}\left(s\right)\La^{\TWD}_{I_b}\left(s\right)\La^{\TWD}_{I_u}\left(s\right) \dd r,\label{eq:out_prob_2d}
\end{align}
where,
\begin{align}
\La^{\TWD}_{I_\ell}\left(s\right) &= \frac{1}{1+ s P_u G_u^2 \sigma_\ell^2},\label{eq:lap_li_2d}\\
\La^{\TWD}_{I_b}\left(s\right) &= \prod_{i\in\{1,2,3,4\}} \exp\left\{-\frac{2\pi\la_{u,b,i}}{\alpha_1-2} \frac{\G_{u,b,i}}{\G_{u,b,1}} F\left(\alpha_1, \frac{\G_{u,b,i}}{\G_{u,b,1}} \tau\right)r^2 \tau \right\},\label{eq:lap_int_b_2d}\\
\La^{\TWD}_{I_u}\left(s\right) &= 2\pi \la\int_0^\infty \!\!\rho \exp\left\{\!-\pi \rho^2\left(\la +\!\! \sum_{i\in\{1,2,3,4\}} \frac{2\la_{u,u,i}}{\alpha_2-2} \G_{u,u,i} F\left(\alpha_2, \frac{s \G_{u,u,i} P_u}{\rho^{\alpha_2}}\right)\frac{sP_u}{\rho^{\alpha_2}}\right)\!\right\}\dd \rho,\label{eq:lap_int_u_2d}
\end{align}
with $s = \frac{\tau r^{\alpha_1}}{P_b G_b G_u}$ and $\tau=2^R-1$.
\end{theorem}

\begin{proof}
See Appendix \ref{prf:out_prob_2d}.
\end{proof}

The main difference between the two architectures is that in the three-node case, the receiver is not subject to any LI due to the HD mode operation. Despite that, the receiver is subject to intra-cell interference. Therefore, the SINR of $u_0$ in the three-node architecture is the same as \eqref{eq:sinr} with $\mathds{1}_{\rm FD} = 0$.

\begin{theorem}\label{thm:out_prob_3d}
The outage probability of a typical receiver in the three-node architecture is
\begin{align}
P_{\THD} =
1 - 2 \pi \la \int_0^\infty r \exp(-\la \pi r^2-s\sigma_n^2)~
\La^{\THD}_{I_b}\left(s\right) \La^{\THD}_{I_u}\left(s\right) \dd r,\label{eq:out_prob_3d}
\end{align}
where,
\begin{align}
\La^{\THD}_{I_b}\left(s\right) &= \prod_{i\in\{1,2,3,4\}}\exp\left\{-\frac{2\pi\la_{u,b,i}}{\alpha_1-2} \frac{\G_{u,b,i}}{\G_{u,b,1}} F\left(\alpha_1, \frac{\G_{u,b,i}}{\G_{u,b,1}} \tau\right)r^2 \tau \right\},\label{eq:lap_int_b_3d}\\
\La^{\THD}_{I_u}\left(s\right) &= \prod_{i\in\{1,2,3,4\}}\exp\left\{-\frac{2\pi^2\la_{u,u,i}}{\alpha_2} \csc\left(\frac{2 \pi}{\alpha_2}\right) \left(s P_u \G_{u,u,i}\right)^\frac{2}{\alpha_2}\right\},\label{eq:lap_int_u_3d}
\end{align}
with $s = \frac{\tau r^{\alpha_1}}{P_b G_b G_u}$ and $\tau=2^R-1$.
\end{theorem}

\begin{proof}
The proof follows similar steps as the proof of Theorem \ref{thm:out_prob_2d} with the main difference being in the evaluation of $\La^{\THD}_{I_u}(s)$. The receiver in the three-node architecture experiences  intra-cell interference from the uplink user in the same cell. Therefore, the limits of the integral in $\La^{\THD}_{I_u}(s)$ are, in this case, from zero to $\infty$ and \eqref{eq:lap_int_u_3d} is derived with the help of \cite[Eq. (3.194.4)]{GRA}. Finally, since $\La^{\THD}_{I_b}\left(s\right) = \La^{\TWD}_{I_b}\left(s\right)$ and $\mathds{1}_{\rm FD} = 0$ the result follows.
\end{proof}

\subsection{Outage probability at the uplink}\label{subsec:out_up}

The analysis for the outage probability at the uplink follows the same steps to above. We turn our attention to the three-node architecture which is of particular interest. We assume that each BS in the three-node architecture employs antenna separation techniques to passively suppress the LI. The level of achievable passive LI suppression is given by \eqref{eq:ang_fun} in Section \ref{subsec:inter}. In this case, the total channel gain $I_\ell$ from the LI at $b_o$ after active cancellation and passive suppression is given by,
\begin{align}\label{eq:li_3u}
I_\ell = P_u G_b^2 |h_\ell|^2 \supfn(B_0 + \g_b (1-B_0)),
\end{align}
where $|h_\ell|^2 \sim \exp\left(1/\sigma^2_\ell\right)$ and $B_0 \sim {\rm Bernoulli}\left(\frac{1}{M_b}\right)$ is a binary random variable with
\begin{align}\label{eq:bern}
B_0 =
\begin{cases}
1 & {\rm with ~prob.~}~ \frac{1}{M_b} \,~~~~(\theta = 0),\\
0 & {\rm with ~prob.~} \frac{M_b-1}{M_b} ~~~(\theta \neq 0),
\end{cases}
\end{align}
since the power gain of the LI signal is $G_b^2$ for $\theta = 0$ and $G_b H_b$ otherwise. Note that in \eqref{eq:li_3u} we consider the active cancellation and passive suppression of the LI separately. However, in reality, the active cancellation mechanism attempts to mitigate the passively suppressed LI and therefore a more ``realistic" model would be to express the variance $\sigma^2_\ell$ as a function of $\supfn$. For the sake of simplicity, we assume that $\supfn$ is a normalization factor of $|h_\ell|^2$ which makes no difference in the final results. We can now state the following theorem.

\begin{theorem}\label{thm:out_prob_3u}
The outage probability at the typical BS in the three-node architecture is,
\begin{align}
P_{\THU} = 1 - 2 \pi \lambda \int_0^\infty r \exp(-\la \pi r^2-s\sigma_n^2)~
\La^{\THU}_{I_\ell}(s) \La^{\THU}_{I_b}(s) \La^{\THU}_{I_u}(s) \dd r,\label{eq:out_prob_3u}
\end{align}
where,
\begin{align}
\La^{\THU}_{I_\ell}(s) &= \frac{1}{M_b}\left[\frac{1}{1+ s P_b G_b^2 \sigma_\ell^2} +\hspace{-5mm}
\mathlarger{\sum}\limits_{\substack{\theta \in [-\pi, \pi) \setminus \{0\}\\\theta \equiv 0 \Mod{\frac{2\pi}{M_b}}}}\!\frac{1}{1 + s P_b G_b H_b \sigma_\ell^2 \exp\left(\cos\left(\theta_{\rm max}\right)-\cos\left(|\theta|-\theta_{\rm max}\right)\right)}\right],\label{eq:lap_li_3u}\\[-3mm]
\La^{\THU}_{I_b}\left(s\right) &= 2\pi \la\int_0^\infty \!\!\rho \exp\left\{\!-\pi \rho^2\left(\la +\! \sum_{i\in\{1,2,3,4\}} \frac{2\la_{b,b,i}}{\alpha_2-2} \G_{b,b,i} F\left(\alpha_2, \frac{s \G_{b,b,i} P_b}{\rho^{\alpha_2}}\right)\frac{sP_b}{\rho^{\alpha_2}}\right)\!\right\}\dd \rho,\label{eq:lap_int_b_3u}\\
\La^{\THU}_{I_u}\left(s\right) &=\!\!\prod_{i\in\{1,2,3,4\}}\!\!\!\exp\left\{\!-\pi\la_{b,u,i}\left(\!\frac{2\pi}{\alpha_1} \csc\left(\frac{2\pi}{\alpha_1}\right) \left(s P_u\G_{b,u,i}\right)^\frac{2}{\alpha_1}\! - \int_0^\infty \!\!\frac{s P_u\G_{b,u,i}e^{-\pi\la z}}{s P_u \G_{b,u,i}+z^\frac{\alpha_1}{2}}\dd z\!\right)\!\right\},\label{eq:lap_int_u_3u}
\end{align}
with $s = \frac{\tau r^{\alpha_1}}{P_u G_b G_u}$ and $\tau = 2^R - 1$.
\end{theorem}

\begin{proof}
See Appendix \ref{prf:out_prob_3u}.
\end{proof}

The outage probability $P_{\TWU}$ of the typical BS in the two-node architecture can be easily derived from Theorem \ref{thm:out_prob_3u} and thus we exclude its representation for brevity. In this case, $\theta = 0$ and $B_0 = 1$, so \eqref{eq:li_3u} gives $I_\ell = P_u G_b^2 |h_\ell|^2$. Hence, $\La^{\TWU}_{I_b}\left(s\right) = \La^{\THU}_{I_b}\left(s\right)$, $\La^{\TWU}_{I_u}\left(s\right) = \La^{\THU}_{I_u}\left(s\right)$ and $\La^{\TWU}_{I_\ell}(s)$ is derived similarly to \eqref{eq:lap_li_2d}.

\subsection{Special cases}\label{subsec:special}
The derived expressions in Theorems \ref{thm:out_prob_2d}, \ref{thm:out_prob_3d} and \ref{thm:out_prob_3u} provide a general result for the outage probability of each scenario. However, due to the complexity of these expressions, it is difficult to gain insight on the behaviour of each scenario. Therefore, in this section, for the sake of reducing notational overhead and deriving closed-form expressions, further assumptions are considered which simplify the model. Specifically, assume that the users and BSs employ the same number of sectorized antennas $M = M_b = M_u$ and let $\g = \g_b = \g_u$. Furthermore, assume that the BSs and the users transmit with the same power, i.e. $P_b = P_u$, and consider high power transmissions which result in an interference-limited network, that is $\sigma_n^2 = 0$. Closed-form expressions for \eqref{eq:out_prob_2d} and \eqref{eq:out_prob_3u} are difficult to derive due to the extra integral in expressions \eqref{eq:lap_int_u_2d} and \eqref{eq:lap_int_b_3u} respectively and therefore we will consider an approximation to facilitate our investigations but also to help us gain insight into the network's behaviour. To approximate $\La^{\TWD}_{I_u}\left(s\right)$ $\left(\La^{\TWU}_{I_b}\left(s\right), \La^{\THU}_{I_b}\left(s\right)\right)$, we will assume that the closest interfering user (BS) is located at a distance at least $r$, i.e., the distance to the user's (BS's) associated BS (user)\footnote{An appropriate scheduling mechanism ensures this distance in order to protect the system from strong co-channel interference.} \cite{LEE}. Also, we will assume that the interference fields $\Psi_i$ in the uplink model are homogeneous with density $\la_{b,u,i}$, i.e. we ignore the integral term in \eqref{eq:lap_int_u_3u}; this provides an upper bound. By letting $\Lambda_i \triangleq \left\{\frac{1}{M^2}, \frac{(M-1)}{M^2}, \frac{(M-1)}{M^2}, \frac{(M-1)^2}{M^2}\right\}$ and $\G_i \triangleq \left\{1, \g, \g, \g^2\right\}$, $i \in \{1, 2, 3, 4\}$, we state the following.

\begin{proposition}\label{prop:out_fd_approx}
The outage probability of a typical FD-mode node is given by,
\begin{align}
P_x = 1 - 2 \pi \la \int_0^\infty r \exp\left(-\mathcal{G}_x \la \pi r^2\right) \La^x_{I_\ell}\left(s\right) \dd r, x \in \{\TWD, \TWU, \THU\},\label{eq:out_fd_approx}
\end{align}
where,
\begin{align}\label{eq:G_approx1}
\mathcal{G}_\TWD = 1 + 2\tau \sum_{i\in\{1,2,3,4\}} \Lambda_i \G_i \left(\frac{F\left(\alpha_1, \G_i \tau \right)}{\alpha_1-2} +  \frac{r^{\alpha_1-\alpha_2}F\left(\alpha_2, r^{\alpha_1-\alpha_2} \G_i \tau\right)}{\alpha_2-2}\right),
\end{align}
and
\begin{align}\label{eq:G_approx2}
\mathcal{G}_x = 1 + 2\sum_{i\in\{1,2,3,4\}} \Lambda_i\left( \frac{\pi\tau^\frac{2}{\alpha_1}}{\alpha_1}\csc\left(\frac{2\pi}{\alpha_1}\right)  \G_i^\frac{2}{\alpha_1} + \frac{\tau r^{\alpha_1-\alpha_2}}{\alpha_2-2} \G_i F\left(\alpha_2, r^{\alpha_1-\alpha_2} \G_i \tau\right)\right),
\end{align}
for $x \in \{\TWU, \THU\}$.
\end{proposition}

\begin{proof}
We first deal with the downlink case. By setting the limits of the integral in \eqref{eq:int_approx} from $r$ to $\infty$, $\La^{\TWD}_{I_u}\left(s\right)$ changes to
\[\La^{\TWD}_{I_u}\left(s\right) = \exp\left\{-\frac{2\tau \la \pi r^{2+\alpha_1-\alpha_2}}{\alpha_2-2} \sum_{i\in\{1,2,3,4\}} \Lambda_i \G_i F\left(\alpha_2, r^{\alpha_1-\alpha_2} \G_i \tau \right)\right\},\]
again based on \cite[Eq. (3.194.2)]{GRA}. Then, the result follows from simple algebraic manipulations. The expressions for the uplink case can be derived in a similar way.
\end{proof}

We now consider the asymptotic case where the number of employed antennas tends to infinity. In this case, the sectorized antennas generate very tight beams and so co-channel interference occurs only from the side-lobes (Case 4 in Section \ref{subsec:inter}). In other words, $\la_{x,y,i} = \la$ for $i = 4$ and $\la_{x,y,i} = 0$ otherwise, $x, y \in \{b, u\}$.

\begin{proposition}\label{prop:out_fd_asym}
The outage probability of a typical FD-mode node in the asymptotic case $M \rightarrow \infty$ is given by,
\begin{align}
P_x^\infty = 1 - 2 \pi \la \int_0^\infty r \exp\left(-\mathcal{G}_x \la \pi r^2\right) \La^x_{I_\ell}\left(s\right) \dd r, x \in \{\TWD, \TWU, \THU\},\label{eq:out_fd_asym}
\end{align}
with,
\begin{align}\label{eq:G_asym1}
\mathcal{G}_\TWD = 1 + 2\tau\g^2 \left(\frac{F\left(\alpha_1, \g^2 \tau \right)}{\alpha_1-2} + \frac{r^{\alpha_1-\alpha_2}F\left(\alpha_2, r^{\alpha_1-\alpha_2} \g^2 \tau\right)}{\alpha_2-2}\right),
\end{align}
\begin{align}\label{eq:G_asym2}
\mathcal{G}_\TWU = \mathcal{G}_\THU = 1 + \frac{2\pi\tau^\frac{2}{\alpha_1}}{\alpha_1}\csc\left(\frac{2\pi}{\alpha_1}\right) \g^\frac{4}{\alpha_1} + \frac{2 r^{\alpha_1-\alpha_2} \tau}{\alpha_2-2} \g^2 F\left(\alpha_2, r^{\alpha_1-\alpha_2} \g^2 \tau\right),
\end{align}
\begin{align}\label{eq:lim_lap_3u}
\lim_{M\rightarrow \infty} \La^{\THU}_{I_\ell}\left(s\right) = \frac{1}{2\pi}\int_{-\pi}^{\pi}\frac{1}{1+\g \sigma_\ell^2 \tau r^{\alpha_1} \supfn}~\dd \theta,
\end{align}
and
\begin{align}\label{eq:lim_lap_2d}
\lim_{M\rightarrow \infty} \La^{\TWD}_{I_\ell}\left(s\right) = \lim_{M\rightarrow \infty} \La^{\TWU}_{I_\ell}\left(s\right) = \frac{1}{1 + \sigma_\ell^2 \tau r^{\alpha_1}}.
\end{align}
\end{proposition}

\begin{proof}
When $M \to \infty$, the interference at a typical node occurs only from the side lobes of the other nodes in the network. In this case, $\Lambda_i = 1$ for $i = 4$ and $\Lambda_i = 0$ otherwise, $1 \leq i \leq 4$, and so \eqref{eq:G_asym1} and \eqref{eq:G_asym2} follow immediately from \eqref{eq:G_approx1} and \eqref{eq:G_approx2}. Moreover, \eqref{eq:lim_lap_3u} follows from the fact that the first term of \eqref{eq:lap_li_3u} converges to zero for  $M \rightarrow \infty$ and the remaining second term is an infinite sum which gives the definite integral. Finally, as $M_b = M_u$, $\La^{\TWD}_{I_\ell}\left(s\right)$ and $\La^{\TWU}_{I_\ell}\left(s\right)$ are independent of $M$ and remain unchanged. 
\end{proof}

The next two propositions deal with the downlink of the three-node architecture. The proofs of these propositions follow similar steps as the ones above and so we omit their inclusion.
\begin{proposition}\label{prop:out_fd_3d}
The outage probability of a typical receiver in the three-node architecture is
\begin{align}
P_\THD = 1 - 2 \pi \la \int_0^\infty r \exp\left(-\mathcal{G}_\THD \la \pi r^2\right) \dd r,
\end{align}
where,
\begin{align}\label{eq:G_3d}
\mathcal{G}_\THD = 1 + 2 \sum_{i\in\{1,2,3,4\}} \Lambda_i \left(\frac{F\left(\alpha_1, \G_i \tau \right)}{\alpha_1-2}\G_i\tau + \frac{\pi r^{\frac{2\alpha_1}{\alpha_2}-2}}{\alpha_2} \csc\left(\frac{2\pi}{\alpha_2}\right)\left(\G_i\tau\right)^{\frac{2}{\alpha_2}}\right).
\end{align}
\end{proposition}

\begin{proposition}\label{prop:asym_3d}
The outage probability of a typical receiver in the three-node architecture when $M \rightarrow \infty$ is given by,
\begin{align}
P_\THD^\infty = 1 \!-\! 2 \pi \la \!\int_0^\infty\!\!\! r \exp\!\left\{\!-\la \pi r^2 \left(1 + 2 \left(\frac{F\left(\alpha_1, \g^2 \tau \right)}{\alpha_1-2}\g^2\tau \!+\! \frac{\pi r^{\frac{2\alpha_1}{\alpha_2}-2}}{\alpha_2} \csc\left(\frac{2\pi}{\alpha_2}\right)\left(\g^2\tau\right)^{\frac{2}{\alpha_2}}\right)\!\right)\!\right\} \dd r.
\end{align}
\end{proposition}

Observe that the passive suppression of the LI in the $\THU$ architecture is improved in the asymptotic case since the probability of no passive suppression becomes zero. On the other hand, the residual LI in the $\TWU$ architecture is not affected which reduces the potential performance gains. Furthermore, as the co-channel interference occurs only from the side-lobes, the performance is highly affected by $\g$. When $\g \rightarrow 0$, then \eqref{eq:G_asym1}, \eqref{eq:G_asym2} and \eqref{eq:lim_lap_3u} tend to one, and so $P_\THU^\infty \to 0$. The same observation can be made for $P_\THD^\infty$. On the other hand, when $\g \rightarrow 0$, $P_x^\infty$ depends entirely on $\La^x_{I_\ell}\left(s\right)$ for $x \in \{\TWD, \TWU\}$. Next, we provide the outage probability for the case $\alpha = \alpha_1 = \alpha_2$ and for specific values of $\alpha$, namely in the region $(2,4]$.

\begin{corollary}\label{cor:alpha4}
The outage probability of a typical FD-mode node when $\alpha = \alpha_1 = \alpha_2 = 4$ is given by,
\begin{align}\label{eq:alpha4}
P_x = 1-\frac{\mathcal{I}_x}{\mathcal{Y}_x\sqrt{\pi}}, x \in \{\TWD, \TWU, \THU\},
\end{align}
for $\sigma_\ell^2 > 0$, and
\begin{align}
P_x = 1-\frac{1}{\mathcal{Y}_x}, x \in \{\TWD, \TWU, \THU\},
\end{align}
for $\sigma_\ell^2 = 0$, where
\begin{align}
\mathcal{Y}_\TWD = 1 + 2\tau \sum_{i\in\{1,2,3,4\}} \Lambda_i \G_i F\left(4, \G_i \tau \right),
\end{align}
\begin{align}
\mathcal{Y}_\TWU = \mathcal{Y}_\THU = 1 + \sum_{i\in\{1,2,3,4\}} \Lambda_i\left( \frac{\pi}{2} \sqrt{\tau\G_i} + \tau \G_i F\left(4, \G_i \tau\right)\right),
\end{align}
\begin{align}
\mathcal{I}_x = G^{1, 3}_{3, 1} \left(\frac{4\sigma_\ell^2\tau}{\left(\pi\la \mathcal{Y}_x\right)^2} ~ \Bigg | ~ \mfrac{0,\frac{1}{2},0}{0} \right), x \in \{\TWD, \TWU\},
\end{align}
and
\begin{align}
\!\mathcal{I}_\THU = \frac{1}{M}\!\left[G^{1, 3}_{3, 1} \left(\frac{4\sigma_\ell^2\tau}{\left(\pi\la \mathcal{Y}_\THU\right)^2} ~ \Bigg | ~ \mfrac{0,\frac{1}{2},0}{0} \right) +\!\!\!
\mathlarger{\sum}\limits_{\substack{\theta \in [-\pi, \pi) \setminus \{0\}\\\theta \equiv 0 \Mod{\frac{2\pi}{M}}}}G^{1, 3}_{3, 1} \left(\frac{4 \sigma_\ell^2 \tau \g \supfn}{\left(\pi\la \mathcal{Y}_\THU\right)^2} ~ \Bigg | ~ \mfrac{0,\frac{1}{2},0}{0} \right)\!\right].
\end{align}
\end{corollary}

\begin{proof}
Using the identity $\displaystyle \frac{1}{1 + cx^k} = G_{1 1}^{1 1} \left(c x^{k} ~ \Big | ~ \mfrac{0}{0} \right)$ and applying the change of variable $r^2 = \upsilon$, the results for $\sigma_\ell^2 > 0$ follow from \cite[Eq. (7.813.2)]{GRA}. The results for $\sigma_\ell^2 = 0$ follow by a simple integration of the exponential function.
\end{proof}

\begin{corollary}\label{cor:alpha3}
Let $\alpha = \alpha_1 = \alpha_2 = \frac{m}{n}$ with $\gcd(m,n)=1$  where $\gcd(m,n)$ is the greatest common divisor of integers $m$ and $n$. Then, the outage probability $P_{\TWD}$ for $2 < \alpha < 4$ is given by,
\begin{align}\label{eq:alpha3}
P_{\TWD} = 1 - \frac{1}{q} \left(\frac{2n m^{\frac{1}{2}}}{(2\pi)^{2n + \frac{m-3}{2}}} G^{2n, 2n+m}_{2n+m, 2n} \left(\left(\frac{m}{\pi \la q}\right)^m \left(\sigma^2_\ell\tau\right)^{2n} ~ \Bigg | ~ \mfrac{\Delta(2n,0), \Delta(m,0)}{\Delta(2n,0)} \right)\right),
\end{align}
where $q = 1+\frac{4\tau}{\alpha-2} \sum_{i\in\{1,2,3,4\}} \Lambda_i \G_i F\left(\alpha, \G_i \tau \right)$ and $\Delta(a,b) = \frac{b}{a},\frac{b+1}{a},\cdots,\frac{b+a-1}{a}$.
\end{corollary}

\begin{proof}
By applying the identity $\displaystyle e^x = G_{0 1}^{1 0} \left(-x ~ \Big | ~ \mfrac{-}{0} \right)$ and using \cite[Eq. (2.24.2.1)]{PB}, \eqref{eq:out_fd_approx} can be solved to yield \eqref{eq:alpha3}.
\end{proof}

\begin{corollary}\label{cor:3d}
The outage probability of a typical receiver in the three-node architecture when $\alpha = \alpha_1 = \alpha_2 = 4$ is
\begin{align}
P_\THD = 1 - \frac{1}{1 + \sum_{i\in\{1,2,3,4\}} \Lambda_i \left(F\left(4, \G_i \tau \right) \G_i\tau + \frac{\pi}{2} \sqrt{\G_i\tau}\right)}.\label{eq:spec_3d}
\end{align}
\end{corollary}

\begin{proof}
By applying the change of variable $r^2 = \upsilon$, the resulting integral gives the result.
\end{proof}

It is clear from Corollary \ref{cor:alpha4} that the performance of a typical FD-mode node in the perfect LI cancellation case ($\sigma_\ell^2 = 0$) is independent of the density $\lambda$. This is also true for the typical receiver in the three-node architecture as can be seen from Corollary \ref{cor:3d}. In particular, this independence is always valid when $\alpha_1 = \alpha_2$ and $\sigma_\ell^2 = 0$. In this case, the expressions \eqref{eq:G_approx1}, \eqref{eq:G_approx2} and \eqref{eq:G_3d} become independent of both $r$ and $\lambda$ and the final expression results from a simple integration of the exponential function (see Corollary \ref{cor:alpha4}). On the other hand, when $\sigma_\ell^2 > 0$, $P_x$, $x \in \{\TWD, \TWU, \THU\}$ does depend on the density $\la$ and thus, in this case, the denser the network the better the outage performance is. This is explained by the fact that the receiver will be closer to its associated BS and consequently the received signal will be improved which will reduce the LI effects. In particular, the distance $r$ is inversely proportional to the density, so when $\la$ becomes very large, $r^\alpha$ converges to $0$. It is clear from \eqref{eq:lap_li_2d} and \eqref{eq:lap_li_3u} that for $r^\alpha \to 0$ then $\La^x_{I_\ell}\left(s\right) \to 1$, $x \in \{\TWD, \TWU, \THU\}$. Hence, as $\la$ becomes large, $P_x$, $x \in \{\TWD, \TWU, \THU\}$ converges to the performance of perfect LI cancellation.

\section{Composite Architecture Network}\label{sec:het}
In this section, we consider the case where the two architectures are employed in the same tier. In other words, we assume that a typical cell employs the two-node architecture with probability $p^{\TWN}$ and three-node architecture with probability $p^{\THN} = 1-p^{\TWN}$. An example of our system model would be the coexistence between FD empowered machine-to-machine (M2M) type users with HD-mode conventional users of cellular/small cell networks forming a heterogeneous network (HetNet) environment such as in 5G \cite{LEE}, \cite{KOUNT}. By the thinning theorem \cite{HAEN1}, the PPP $\Psi$ is split into two smaller independent PPPs which we denote by $\Psi^{\TWN}$ and $\Psi^{\THN}$ with densities $\la^{\TWN} = p^{\TWN}\la$ and $\la^{\THN} = p^{\THN}\la$, respectively. The same applies for the thinning of the PPP $\Phi$. In order to model possible traffic asymmetry between uplink and downlink directions, we assume that the FD-mode users operate $p_u\%$ of the time in bidirectional FD-mode and $(1-p_u)\%$ of the time in downlink HD-mode. Therefore, the FD-mode users transmitting in the uplink at each time slot form an independent PPP with density $p_up^{\TWN}\la$.

In the next subsections, we will study the performance at both the downlink and uplink of this type of composite network when $M \rightarrow \infty$ together with the assumptions given in Section \ref{subsec:special}. Furthermore, we will evaluate the optimal value of $p^{\TWN}$ for the success probability of the uplink and downlink but also for the network throughput with respect to $p_u$ and $\sigma_\ell^2$. Recall that the assumption involving the distance to the closest interfering terminal leads to an upper bound for the success probability when $M > 1$.

\subsection{Performance at the downlink}
A typical receiver in the composite architecture scenario will experience the same aggregate interference from the BSs, regardless of whether the user operates in FD or HD mode, since the BSs of both architectures interfere with the user in a similar way. Hence, the Laplace transforms $\La^{\TWD}_{I_b}(s)$ and $\La^{\THD}_{I_b}(s)$ for the BS-interference in a composite architecture downlink scenario are still given by \eqref{eq:lap_int_b_2d} and \eqref{eq:lap_int_b_3d} respectively.

The Laplace transform for the interference experienced at a typical FD-mode user from the other users is given by,
\begin{align}
	\La^{\TWD}_{I_u}\left(s\right) = \exp\left\{-\frac{2\pi\la\left(p_u p^{\TWN} + p^{\THN}\right)\g^2}{\alpha-2} F\left(\alpha, \g^2\tau\right)r^2 \tau\right\}\label{eq:lap_int_u_2d_multi}
\end{align}
and at a typical HD-mode user,
\begin{align}
	\La^{\THD}_{I_u}\left(s\right) = \exp\left\{-\frac{2\pi^2\la\left(p_u p^{\TWN} + p^{\THN}\right)}{\alpha} \csc\left(\frac{2 \pi}{\alpha}\right)r^2 (\tau\g^2)^{2/\alpha}\right\}\label{eq:lap_int_u_3d_multi}
\end{align}

In what follows, we provide the outage probability of an FD and an HD-mode user in a composite architecture scenario. We state the results without proof as they are extensions of Propositions \ref{prop:out_fd_asym} and \ref{prop:asym_3d}.

\begin{figure}[t]
  \begin{subfigure}{0.5\linewidth}
    \centering
    \includegraphics[width=0.8\linewidth]{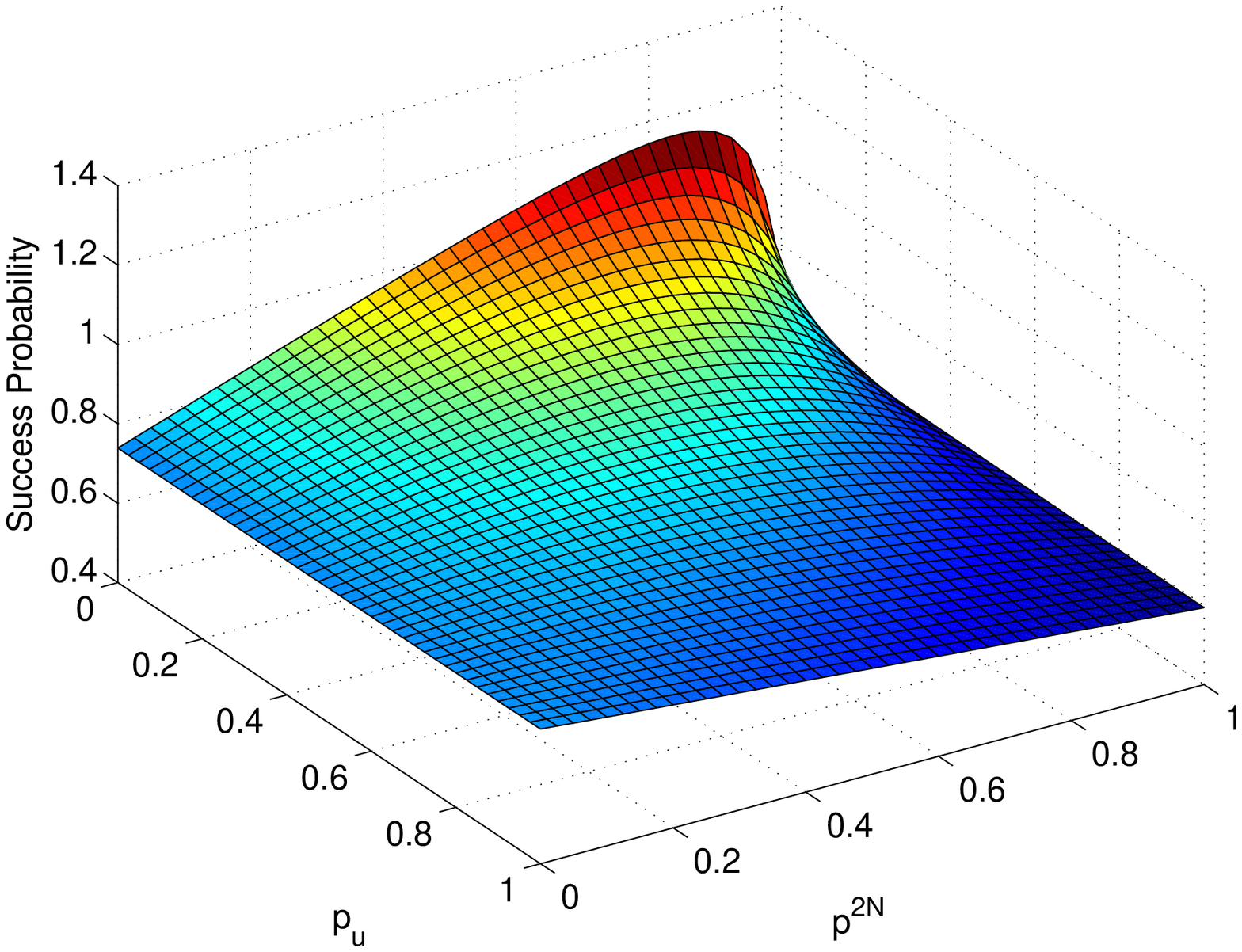}
    \caption{$\sigma_\ell^2 = -30$ dB.}
    \label{fig:success dl LLI}
  \end{subfigure}
  \begin{subfigure}{0.5\linewidth}
    \centering
    \includegraphics[width=0.8\linewidth]{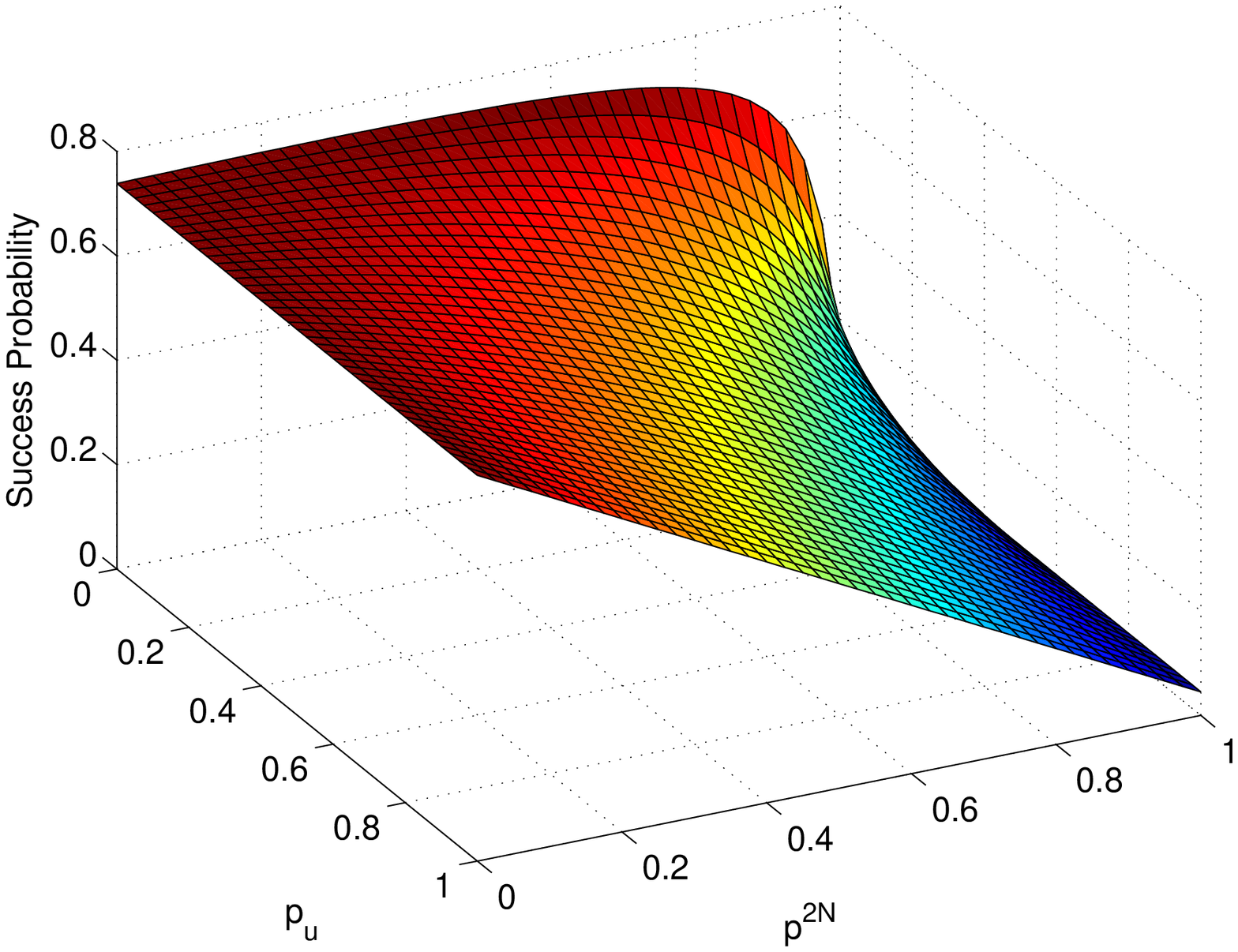}
    \caption{$\sigma_\ell^2 = 0$ dB.}
    \label{fig:success dl HLI}
  \end{subfigure}
  \caption{Success probability at downlink vs $p_u$ and $p^{\TWN}$; $\alpha=4$, $R = 1$ bpcu, $\g = 0.2$, $\la = 10^{-2}$.}\vspace{-5mm}
\end{figure}

\begin{proposition}\label{prop:multi_out_2d}
The outage probability of a typical FD-mode receiver in an FD composite architecture scenario is given by Proposition \ref{prop:out_fd_asym} with
\begin{align}\label{G_comp_down}
	\mathcal{G}_\TWD = 1 + \frac{2 \tau \g^2}{\alpha-2} F\left(\alpha, \g^2 \tau \right)\left(p_u p^{\TWN} + p^{\THN} + 1\right).
\end{align}
\end{proposition}

\begin{proposition}\label{prop:multi_out_3d}
The outage probability of a typical HD-mode receiver in an FD composite architecture scenario is given by,
\begin{align}\label{eq:multi_out_3d}
	P'_{\THD} = 1 - \frac{\alpha(\alpha-2)}{\left(\alpha-2\right)\left(\alpha + 2 \pi (\tau\g^2)^{\frac{2}{\alpha}} \csc\left(\frac{2 \pi}{\alpha}\right)\left(p_u p^{\TWN} + p^{\THN}\right)\right) + 2 \alpha \g^2\tau F\left(\alpha, \g^2\tau\right)}.
\end{align}
\end{proposition}

Given the above two propositions, we can now state the following.

\begin{proposition}\label{prop:out_multi_down}
The outage probability of a typical receiver in an FD composite architecture scenario is given by,
\begin{align}\label{eq:out_multi_down}
  \Pi_d = p^{\TWN}P'_{\TWD} + p^{\THN}P'_{\THD},
\end{align}
where $P'_{\TWD}$ and $P'_{\THD}$ are given in Propositions \ref{prop:multi_out_2d} and \ref{prop:multi_out_3d} respectively.
\end{proposition}

\begin{proof}
See Appendix \ref{prf:out_multi_down}.
\end{proof}

Note that when $p_u = 1$ then \eqref{eq:out_multi_down} becomes $\Pi_d = p^{\TWN}P_{\TWD} + p^{\THN}P_{\THD}$ where $P_{\TWD}$ and $P_{\THD}$ are given by Propositions \ref{prop:out_fd_asym} and \ref{prop:asym_3d} respectively. The cases with specific values of $\alpha$ can be easily derived in a similar manner as above and so they are excluded.

From Proposition \ref{prop:out_multi_down} we can see that the outage probability of the typical receiver depends on the HD-mode and FD-mode user densities, FD-mode user traffic, and also the LI cancellation capability of the system. Therefore, the optimal portion of the FD-mode users that maximizes the success probability of the downlink can be obtained as,
\begin{align}\label{eq:optimization DL}
	& p^{\TWN*} = \argmax_{p^{\TWN}} ~ (1-\Pi_d)\nonumber\\[-2.5mm]
	& \text{subject to}\qquad 0 \leq p^{\TWN}\leq 1.
\end{align}

The optimization problem \eqref{eq:optimization DL} is nonconvex and a globally optimal solution is difficult to obtain. In order to tackle this problem, we can resort to numerical methods, such as the projected gradient algorithm (PGA), to find a locally optimal solution. The advantage of PGA is that it only requires the evaluation of the first-order derivative of the objective function while other approaches for nonlinear programming, such as the sequential quadratic programming and the Gauss-Newton method, also require the evaluation of the second-order derivative \cite{DPB}. Since \eqref{eq:alpha4} contains Meijer G-functions, the complexity of computing the second-order derivative for solving \eqref{eq:optimization DL} is very high.

Figs. \ref{fig:success dl LLI} and \ref{fig:success dl HLI} show the success probability as a function of $p^{\TWN}$ and $p_u$ for $\sigma_\ell^2 = -30$ dB and $\sigma_\ell^2 = 0$ dB respectively and with $\alpha = 4$. We see that when the LI cancellation is imperfect and $\sigma_\ell^2 = 0$ dB, the three-node architecture (HD-mode users) is preferred. Nevertheless, when $p_u$ is decreased and particularly for values $p_u < 0.5$, a composite architecture can be used to enhance the success probability and consequently the downlink throughput. On the other hand, when the residual LI gain is negligible, i.e. when $\sigma_\ell^2 = -30$ dB, a composite architecture is preferred again for small values of $p_u$. This is expected since for large values of $p_u$ the residual LI will degrade the overall performance of the network and thus in this case the three-node architecture is preferred.

\begin{figure}[t]
	\begin{subfigure}{0.5\linewidth}\centering
		\includegraphics[width=0.8\linewidth]{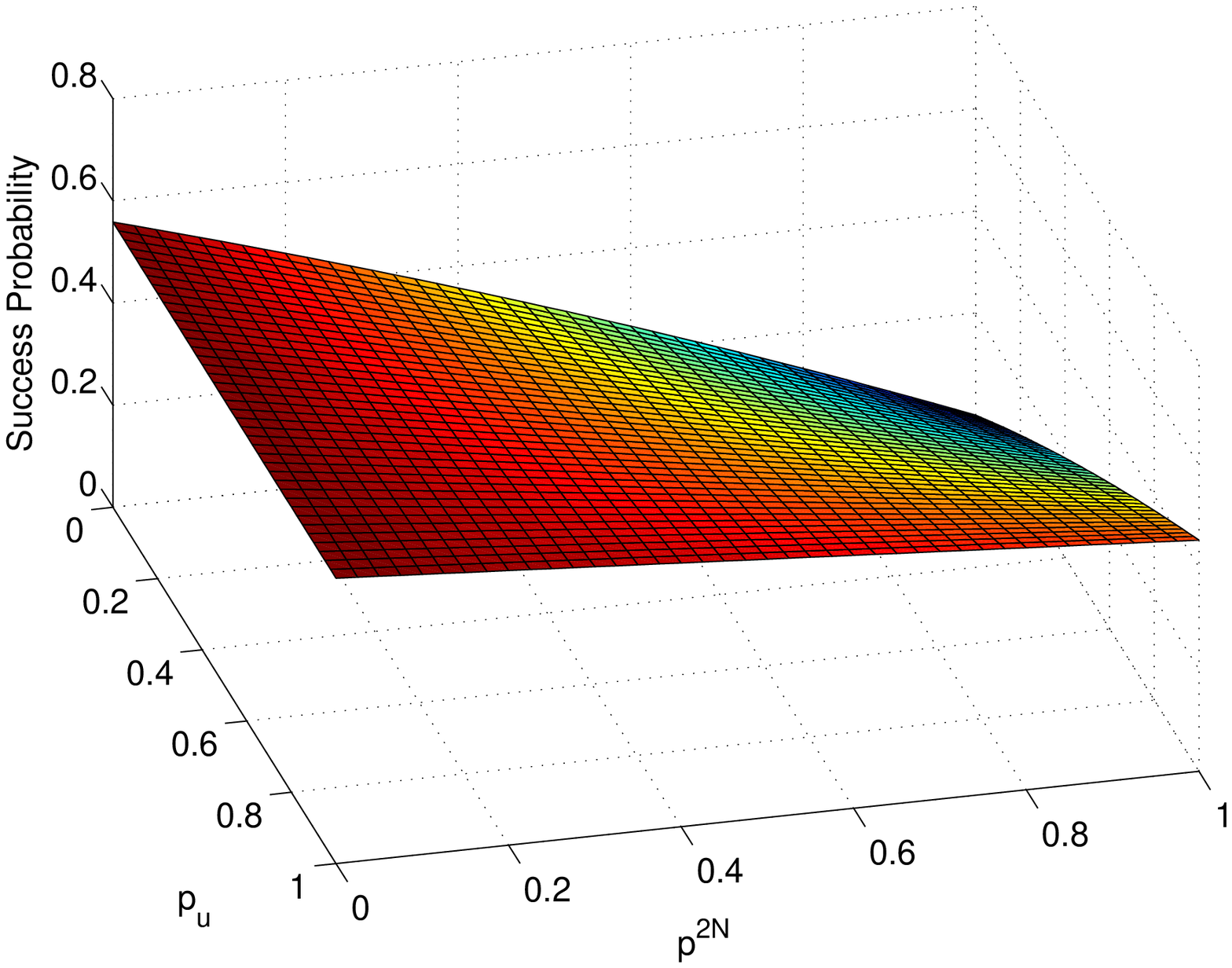}
		\caption{$\sigma_\ell^2 = -30$ dB.}
		\label{fig:success ul LLI}
	\end{subfigure}
	\begin{subfigure}{0.5\linewidth}\centering
		\includegraphics[width=0.8\linewidth]{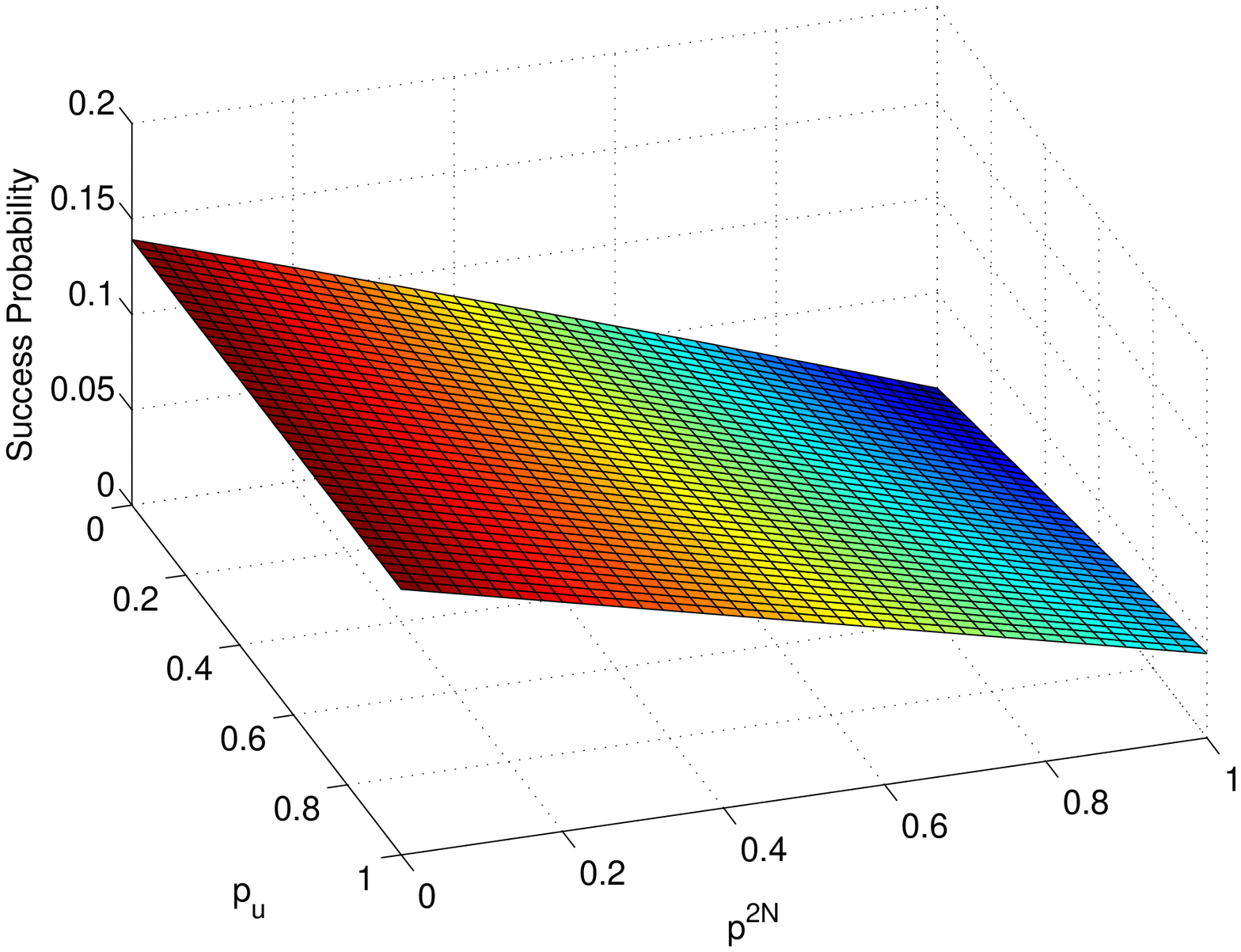}
		\caption{$\sigma_\ell^2 = 0$ dB.}
		\label{fig:success ul HLI}
	\end{subfigure}
	\caption{Success probability at uplink vs $p_u$ and $p^{\TWN}$; $\alpha_1=4$, $\alpha_2=3$, $R = 1$ bpcu, $\g = 0.2$, $\la = 10^{-2}$.}\vspace{-8mm}
\end{figure}

\subsection{Performance at the uplink}
A typical uplink BS from either architecture experiences co-channel interference from outside its cell. Therefore, under this section's assumptions, the Laplace transforms $\La^x_{I_b}\left(s\right)$, $x \in \{\TWU, \THU\}$, are obtained similarly to expression \eqref{eq:lap_int_b_2d}, as the nearest interfering BS is assumed to be at a distance $r$. Also, in a similar manner, $\La^x_{I_u}\left(s\right)$, $x \in \{\TWU, \THU\}$ are derived as expression \eqref{eq:lap_int_u_3d}. Hence, we have the following for the uplink case.

\begin{proposition}\label{prop:out_multi_up}
The outage probability of a typical uplink BS in an FD composite architecture scenario is given by,
\begin{align}\label{eq:out_multi_up}
\Pi_u = p^{\TWN}P'_{\TWU} + p^{\THN}P'_{\THU},
\end{align}
where both $P'_{\TWU}$ and $P'_{\THU}$ are given in Proposition \ref{prop:out_fd_asym} with
\begin{align}\label{G_comp_up}
\mathcal{G}_\TWU = \mathcal{G}_\THU = 1 + \frac{2\pi\tau^\frac{2}{\alpha_1}}{\alpha_1}\csc\left(\frac{2\pi}{\alpha_1}\right) \g^\frac{4}{\alpha_1}\left(p_up^{\TWN} + p^{\THN}\right) + \frac{2 r^{\alpha_1-\alpha_2} \tau}{\alpha_2-2} \g^2 F\left(\alpha_2, r^{\alpha_1-\alpha_2} \g^2 \tau\right).
\end{align}
\end{proposition}
The optimal $p^{\TWN}$, maximizing the uplink success probability, could be obtained by solving the following optimization
\begin{align}\label{eq:optimization UL}
	& p^{\TWN*} = \argmax_{p^{\TWN}} ~ (1-\Pi_u)\nonumber\\[-2.5mm]
	& \text{subject to}\qquad 0 \leq p^{\TWN}\leq 1.
\end{align}

Similarly to the downlink case, given the outage expression in Proposition \ref{prop:out_multi_up}, the optimization problem in \eqref{eq:optimization UL} does not admit a closed-form solution and therefore the optimal $p^{\TWN*}$ is efficiently solved via numerical calculation. Figs. \ref{fig:success ul LLI} and \ref{fig:success ul HLI} show the success probability as a function of $p^{\TWN}$ and $p_u$ for $\sigma_\ell^2 = -30$ dB and $\sigma_\ell^2 = 0$ dB respectively and with $\alpha_1 = 4$, $\alpha_2 = 3$. It is clear that uplink transmissions are more susceptible to the residual LI strength than the downlink ones, since all uplink transmissions are almost in outage for $\sigma_\ell^2 = 0$ dB.  This observation simply means that the LI cancellation mechanism at the BSs should be more effective than the one at the FD-mode users. Furthermore, unlike the downlink case, when $\sigma_\ell^2 = -30$ dB a composite architecture is preferred for $p_u>0.5$.

\begin{figure}[t]
	\begin{subfigure}{0.5\linewidth}
		\centering
		\includegraphics[width=0.8\linewidth]{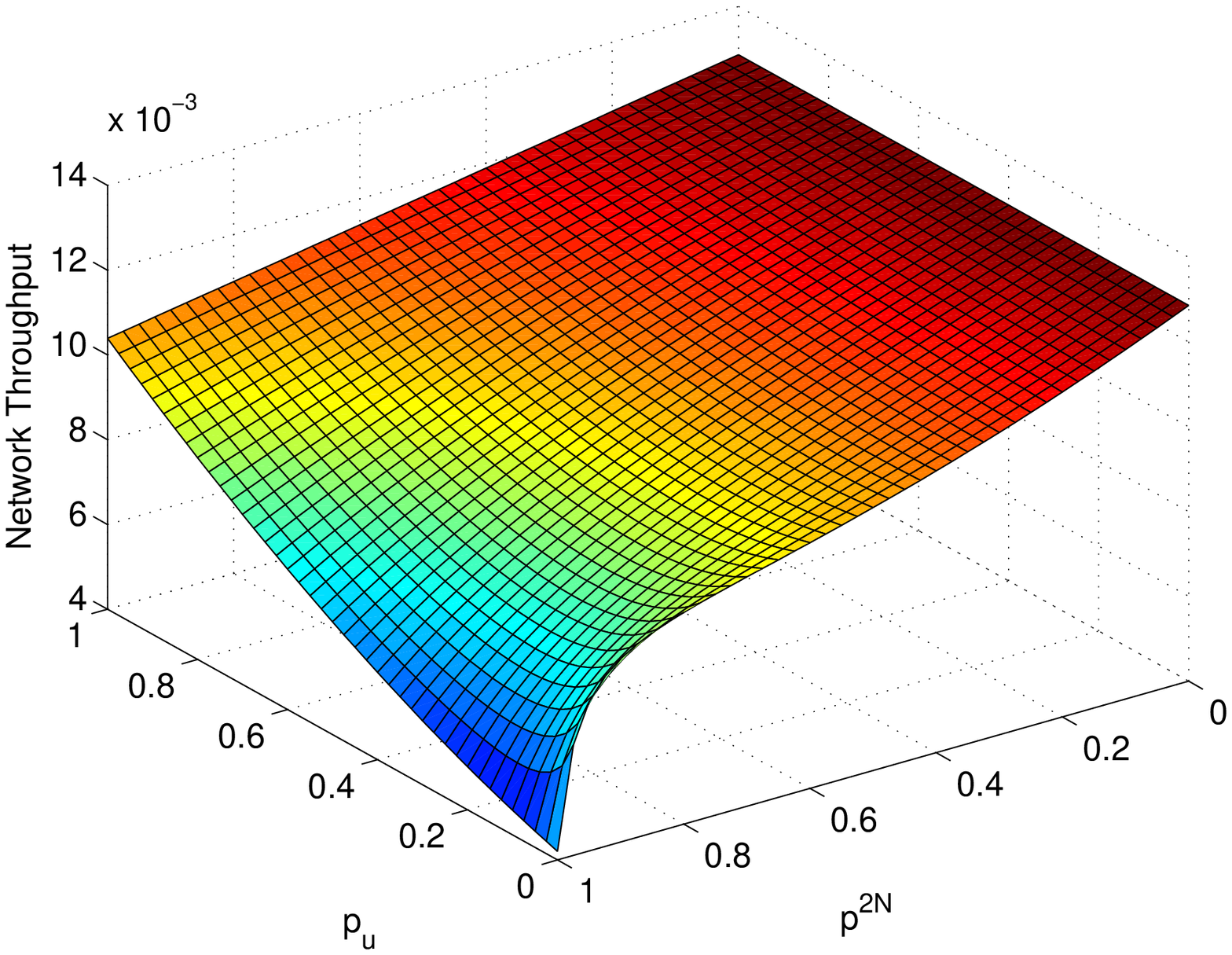}
		\caption{$\la = 10^{-2}$.}
		\label{fig:success net LLam}
	\end{subfigure}
	\begin{subfigure}{0.5\linewidth}
		\centering
		\includegraphics[width=0.8\linewidth]{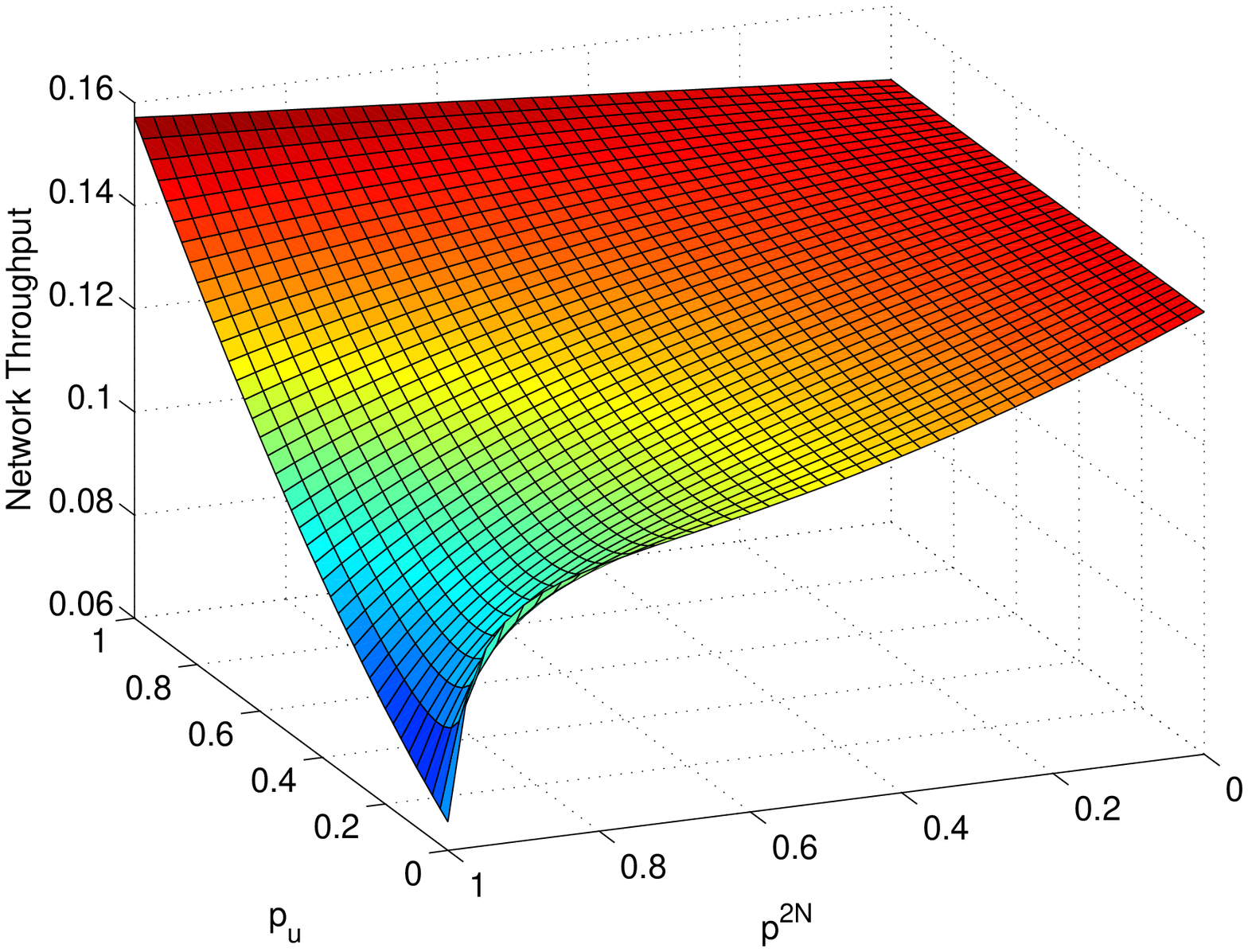}
		\caption{$\la = 10^{-1}$.}
		\label{fig:success net HLam}
	\end{subfigure}
	\caption{Network throughput vs $p_u$ and $p^{\TWN}$; $R = 1$ bpcu, $\g = 0.2$, $\sigma_\ell^2 = -30$ dB.}\vspace{-5mm}
\end{figure}

\subsection{Network throughput}
The implementation of FD-mode can potentially double the throughput of a network compared to HD-mode and hence it is a key metric for the evaluation of the network's performance. The network throughput is defined as the product of the success probability and the sum rate per unit area. When uplink and downlink independent data streams are sent on each time slot, the composite architecture throughput is given by
\begin{align}\label{eq:Outage:UL}
	T(\lambda,\tau,p^{\TWN},p^{\THN}) =\lambda (1-\Pi_d)\log_2(1+\tau) + \lambda(p_up^{\TWN} + p^{\THN}) (1-\Pi_u)\log_2(1+\tau).
\end{align}

With the assumptions from Section \ref{subsec:special} and for $p_u = 1$ we obtain the network throughput as follows,
\begin{align}
	T(\lambda,\tau,p^{\TWN}) =&
	\lambda \log_2(1+\tau)
	\Bigg(\frac{2\left(1-p^{\TWN}\right)}{1 + \frac{4 \pi}{\alpha} (\tau\g^2)^{\frac{2}{\alpha}} \csc\left(\frac{2 \pi}{\alpha}\right) +  \mathcal{G}_\TWD} +\! 2 \pi \la p^{\TWN} \int_0^\infty r e^{-\mathcal{G}_\TWD \pi \la r^2} \La^{\TWD}_{I_\ell}\left(s\right) \dd r\nonumber\\
	&+ 2 \pi \la \int_0^\infty r e^{-\mathcal{G}_\TWU \pi \la r^2} \left( p^{\TWN}\La^{\TWU}_{I_\ell}\left(s\right) + (1-p^{\TWN}) \La^{\THU}_{I_\ell}\left(s\right)\right) \dd r \Bigg),\label{eq:Throughput}
\end{align}
where $\mathcal{G}_d$ and $\mathcal{G}_u$ are given by \eqref{G_comp_down} and \eqref{G_comp_up} respectively.

The optimal $p^{2\mathsf{N}}$ could be obtained by solving the following optimization
\begin{align}\label{eq:optimization}
	&p^{\TWN*} = \argmax_{p^{\TWN}} ~ T(\lambda,\tau,p^{\TWN})\nonumber\\[-2.5mm]
	& \text{subject to}\qquad 0 \leq p^{\TWN}\leq 1.
\end{align}
The above optimization can be solved analytically  and we have the following key result,
\begin{align}\label{eq:network throughput optimization}
	p^{\TWN*}& =
	\begin{cases}
		1,
		& 2 \pi \la \int_0^\infty r\left(\frac{\La^{\TWD}_{I_\ell}\left(s\right)}{e^{\mathcal{G}_\TWD \pi \la r^2}} + \frac{\La^{\TWU}_{I_\ell}\left(s\right) - \La^{\THU}_{I_\ell}\left(s\right)}{e^{\mathcal{G}_\TWU \pi \la r^2}}\right) dr
		> \frac{2}{1 + \frac{4 \pi}{\alpha} (\tau\g^2)^{\frac{2}{\alpha}} \csc\left(\frac{2 \pi}{\alpha}\right) +  \mathcal{G}_\TWD},\\
		0 , & \text{otherwise}.
	\end{cases}
\end{align}

\begin{proof}
The objective function in \eqref{eq:optimization} is clearly a linear function of $p^{\TWN}$ and hence, the optimum solution is located on the boundaries of the region $\mathcal{C} = [0, 1]$, depending on the sign of the first order derivative of the objective function.
\end{proof}

Figs. \ref{fig:success net LLam} and \ref{fig:success net HLam} illustrate the network throughput $T$ as a function of $p^{\TWN}$ and $p_u$ for $\lambda = 10^{-2}$ and $\lambda = 10^{-1}$ respectively, with $\alpha = \alpha_1 = \alpha_2 = 4$ for the downlink, $\alpha_1 = 4$, $\alpha_2 = 3$ for the uplink and $\sigma_\ell^2 = -30$ dB. From this figures, we can see that, when $p_u=1$, the maximum $T$ can be achieved by operating all users in FD-mode (for $\lambda = 10^{-1}$) and  HD-mode (for $\lambda = 10^{-2}$) which confirms the correctness of~\eqref{eq:network throughput optimization}.

\begin{figure}[t]
  \begin{subfigure}{0.45\linewidth}\centering
    \includegraphics[width=0.8\linewidth]{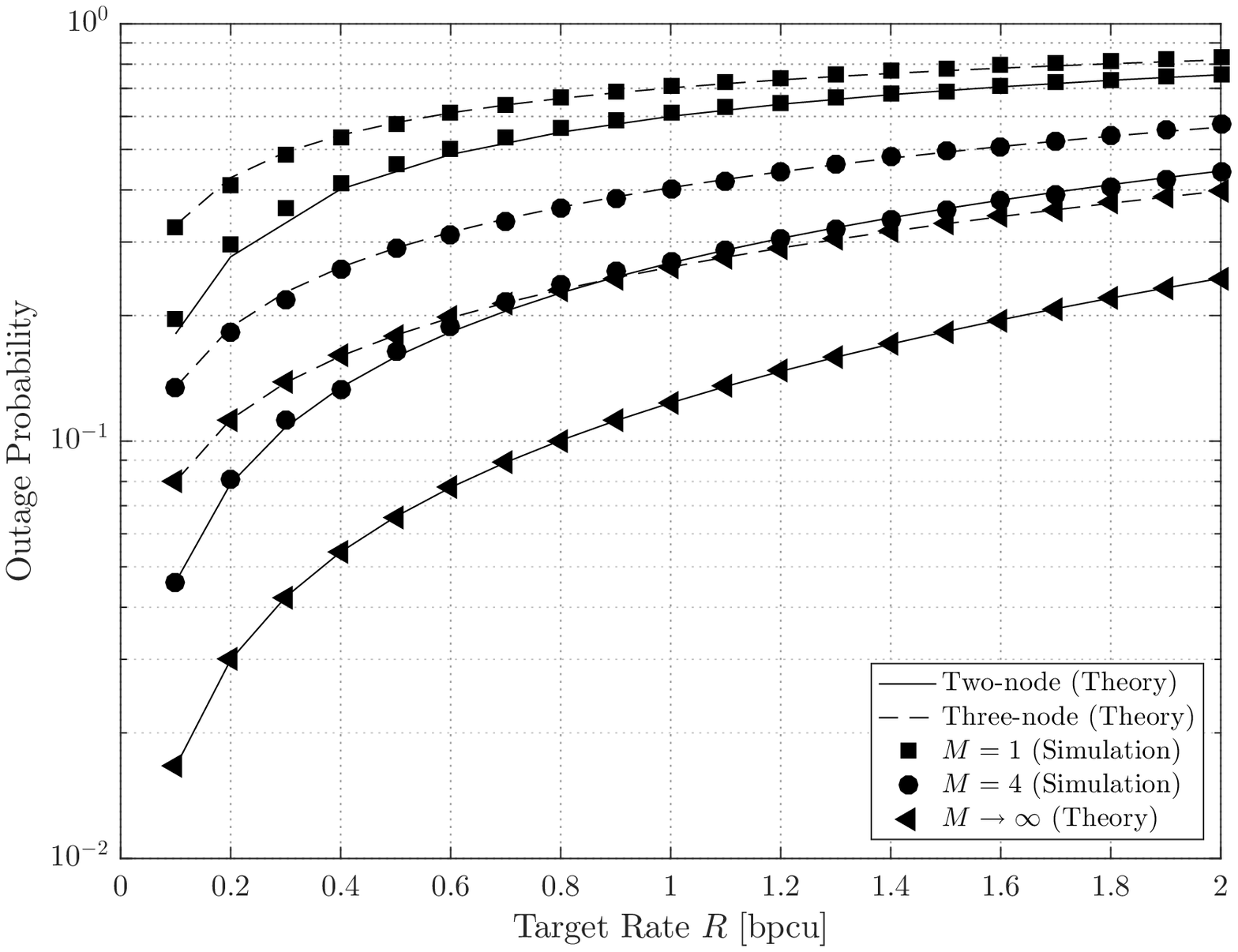}
    \caption{$\sigma_\ell^2 = 0~ (-\infty$ dB)}
    \label{fig:outage_vs_bpcu_0_down}
  \end{subfigure}\hfill
  \begin{subfigure}{0.45\linewidth}\centering
    \includegraphics[width=0.8\linewidth]{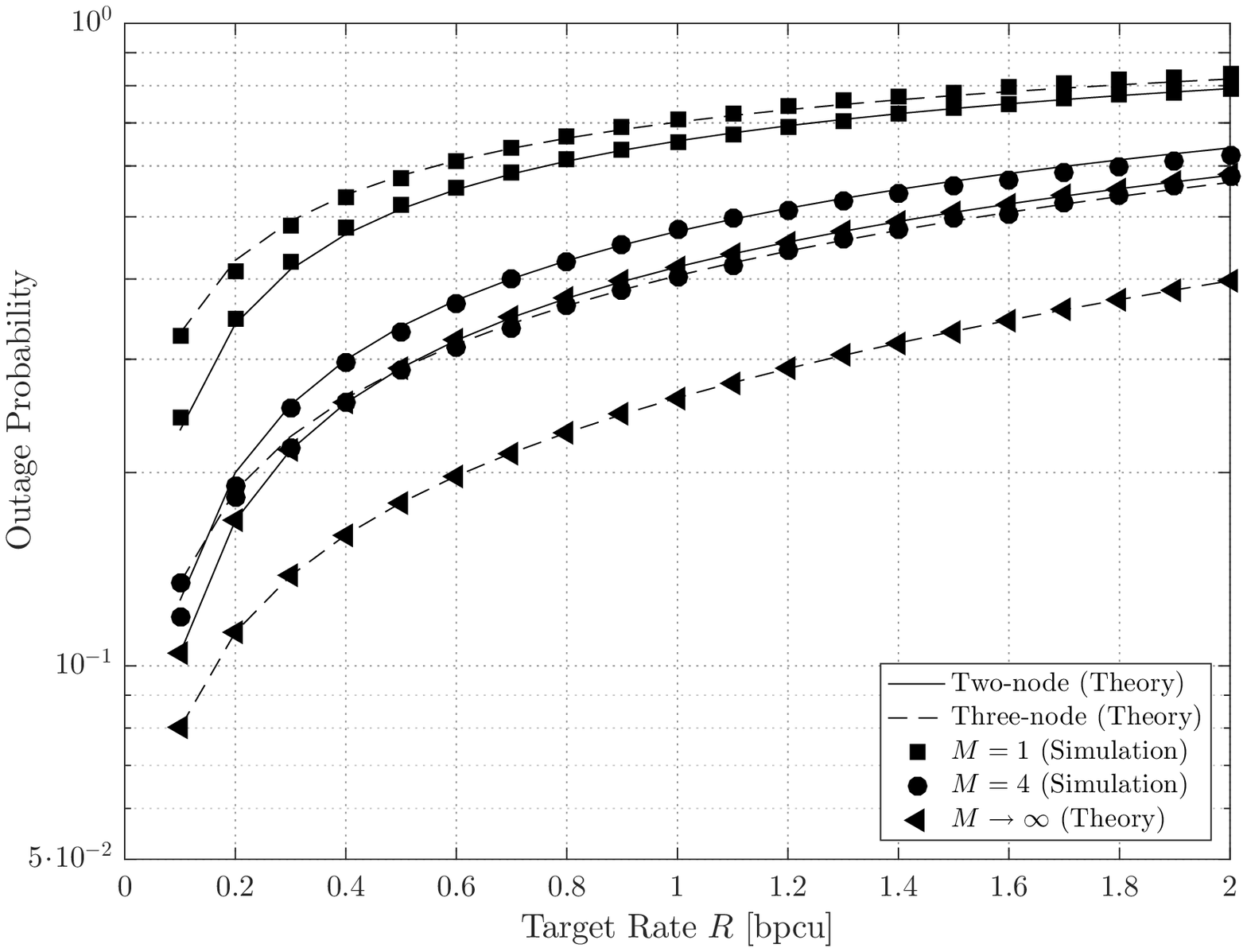}
    \caption{$\sigma_\ell^2 = -30$ dB}
    \label{fig:outage_vs_bpcu_-30db_down}
  \end{subfigure}\vspace{-3mm}
  \caption{Outage probability of downlink vs target rate $R$.}\vspace{-5mm}
\end{figure}

\section{Numerical Results}\label{sec:num}
In this section, we validate the derived expressions and evaluate the proposed model's performance. Unless otherwise stated, the results use the following parameters and assumptions: $\lambda = 10^{-2}$, $\g = \g_b = \g_u = 0.2$, $\sigma^2 = 0$, $M = M_b = M_u$ and $P_b = P_u$. Furthermore, we consider $\alpha = \alpha_1 = \alpha_2 = 4$ for the downlink and $\alpha_1 = 4$ and $\alpha_2 = 3$ for the uplink. The simulation area has a fixed radius of 500 km and the numerical results are obtained by averaging over 10 thousand realizations. The initial density of the users is large enough so that each BS serves, on average, one user. Moreover, we adopt the results in \cite{EE1} and assume that the maximum suppression is achieved at $\theta_{\rm max} = \frac{2\pi}{3}$. In the figures provided, the analytical results are depicted with dashed or solid lines and the simulations with markers except for the asymptotic cases where only analytical results are given. Note that the case $M = 1$ in all figures refers to omnidirectional antennas.

Figs. \ref{fig:outage_vs_bpcu_0_down} and \ref{fig:outage_vs_bpcu_-30db_down} depict the outage probability at the downlink with $\sigma_\ell^2 = 0$ and $\sigma_\ell^2 = -30$ dB respectively for both architectures where the dashed lines represent the analytical results and the dots the simulation results. As expected, the performance of both architectures improves with the employment of directional antennas. Furthermore, the perfect LI cancellation case clearly illustrates the significant gains that the FD radio can potentially provide. However, it is obvious from Fig. \ref{fig:outage_vs_bpcu_-30db_down} that the user in the three-node architecture outperforms the one in the two-node when $M > 1$. This is explained by the fact that the residual LI at the user in the two-node architecture is not affected by the number of directional antennas and so dominates the interference at the user which restrict its performance. The good agreement between the theoretical curves (dashed lines) and the simulation results (markers) validates our mathematical analysis.

\begin{figure}[t]
  \begin{subfigure}{0.32\linewidth}
    \includegraphics[width=\linewidth]{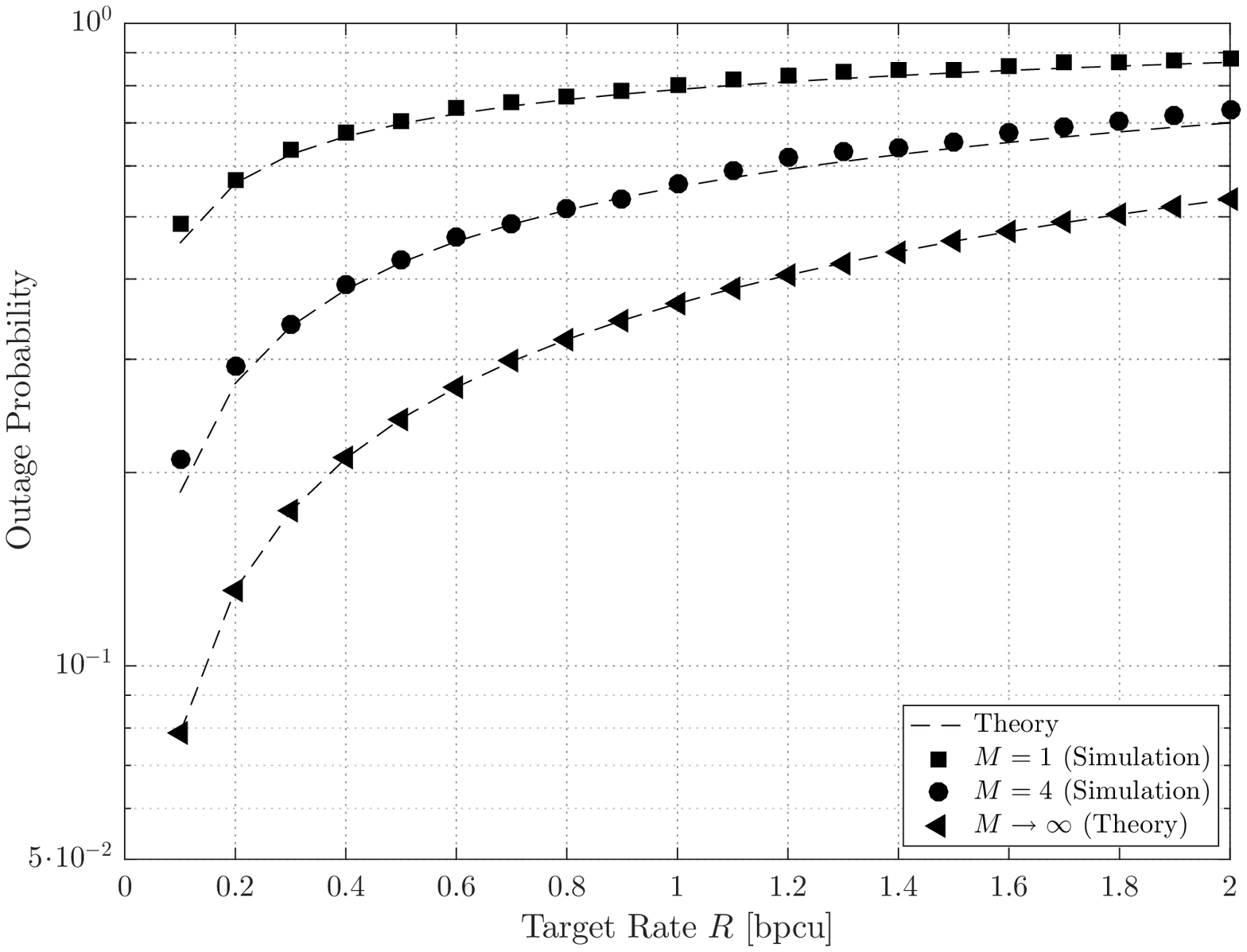}
    \caption{$\sigma_\ell^2 = 0~ (-\infty$ dB)}
    \label{fig:outage_vs_bpcu_0_up}
  \end{subfigure}
  \begin{subfigure}{0.33\linewidth}
    \includegraphics[width=\linewidth]{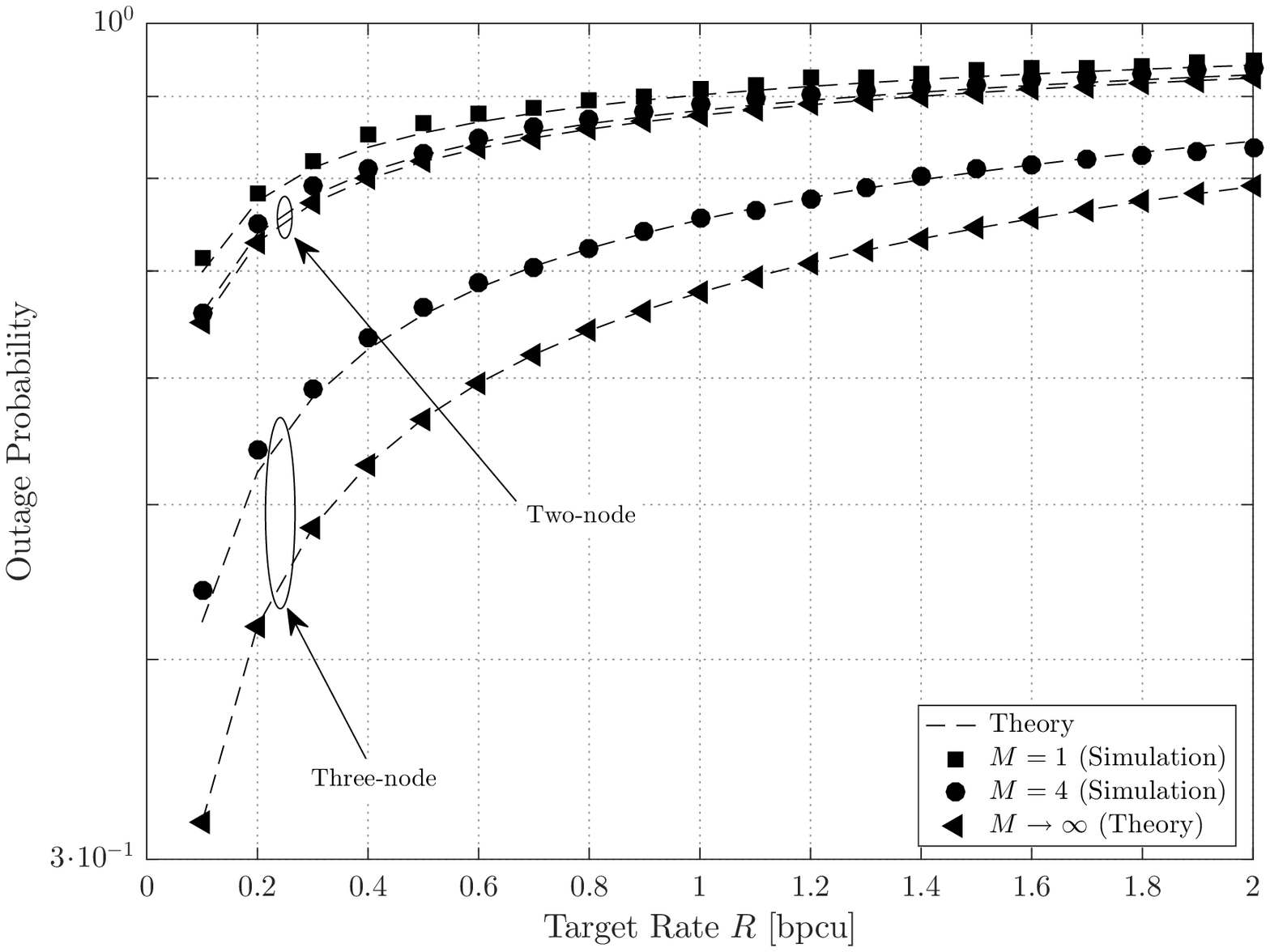}
    \caption{$\sigma_\ell^2 = -10$ dB}
    \label{fig:outage_vs_bpcu_-10db_up}
  \end{subfigure}
  \begin{subfigure}{0.33\linewidth}
    \includegraphics[width=\linewidth]{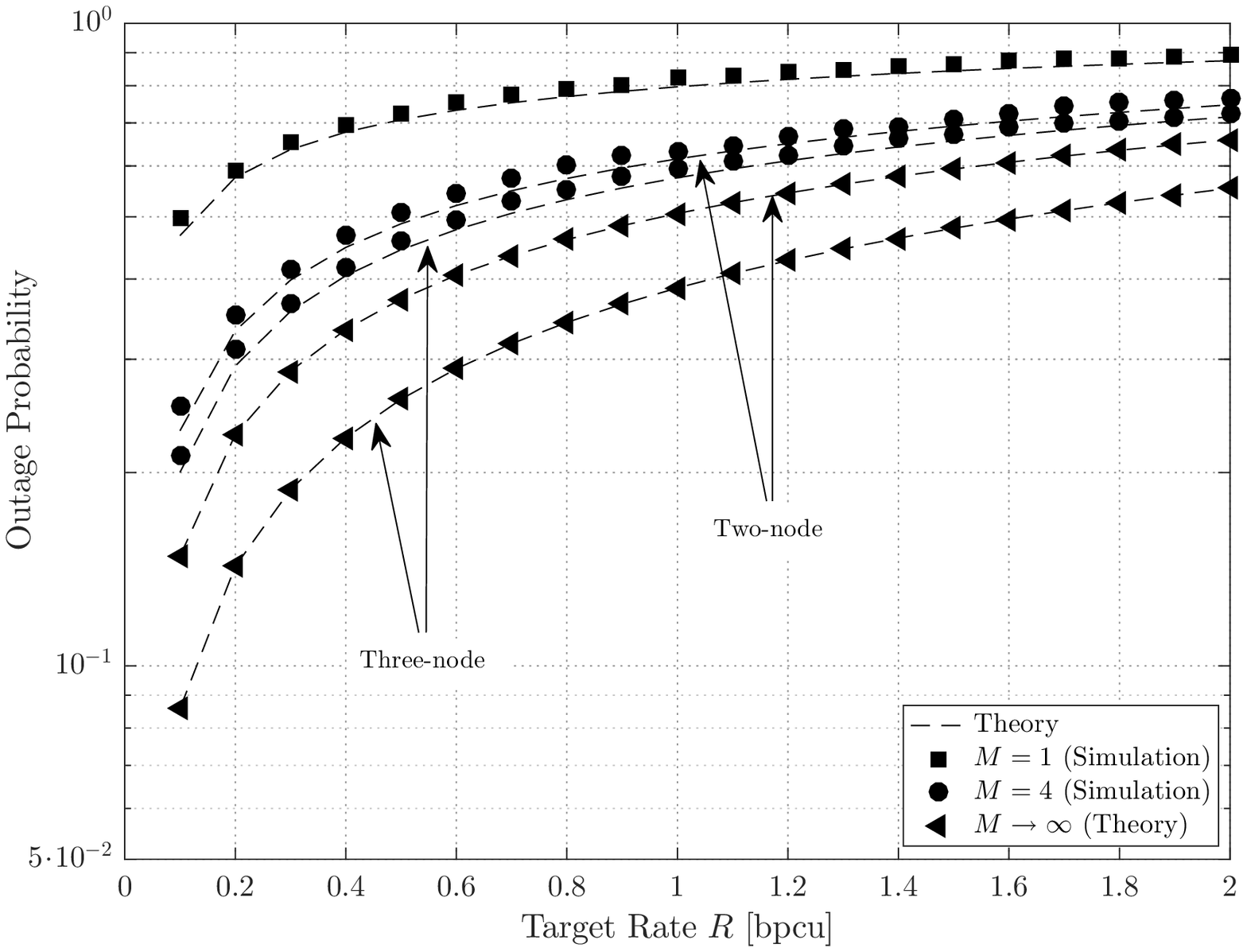}
    \caption{$\sigma_\ell^2 = -30$ dB}
    \label{fig:outage_vs_bpcu_-30db_up}
  \end{subfigure}\vspace{-0.3cm}
  \caption{Outage probability of uplink vs target rate $R$.}\vspace{-5mm}
\end{figure}

Figs. \ref{fig:outage_vs_bpcu_0_up} - \ref{fig:outage_vs_bpcu_-30db_up} illustrate the outage expressions for the uplink. Specifically, Fig. \ref{fig:outage_vs_bpcu_0_up} shows the performance under perfect LI cancellation. In this case, the performance is the same for both architectures since the BSs of the two architectures differ only in the way they handle the LI. Figs. \ref{fig:outage_vs_bpcu_-10db_up} and \ref{fig:outage_vs_bpcu_-30db_up} depict the outage with $\sigma_\ell^2 = -10$ dB and $\sigma_\ell^2 = -30$ dB respectively. It is clear that the BS in the two-node architecture finds it hard to keep up with the three-node one. The same reason applies as in the downlink case. Also, note that for the case $\sigma_\ell^2 = -10$ dB the performance of the BS in the two-node architecture achieves nearly zero improvement from $M=1$ to $M \rightarrow \infty$ since the effect of the residual LI is independent of $M$ (see Proposition \ref{prop:out_fd_asym}). On the other hand, the BS in the three-node architecture achieves a better performance due to the BS's ability to passively suppress the LI. In fact, the passive suppression becomes more efficient with the employment of more antennas.

We show the benefits from the passive suppression method in Fig. \ref{fig:outage_vs_sigma} which illustrates the performance of an FD node in terms of the outage probability, with and without passive suppression, for different values of $\sigma_\ell^2$. In the two extreme cases, $\sigma_\ell^2 \rightarrow -\infty$ dB and $\sigma_\ell^2 \rightarrow \infty$ dB, the two methods have the same performance when $M \to \infty$. In the former case, the outage converges to perfect LI cancellation performance and in the latter case the outage converges to $1$. However, for moderate values, passive suppression provides significant gains, e.g., for $\sigma_\ell^2 = -20$ dB it achieves about $40\%$ reduction. Also, to verify what we said earlier, when an FD node is unable to employ passive suppression techniques, directional antennas become beneficiary only for small values of $\sigma_\ell^2$. Indeed, for values $\sigma_\ell^2 \geq -5$ dB, the performance of active cancellation is the same for any number of antennas. Finally, Fig. \ref{fig:outage_vs_density} shows the effect of the density in the performance of FD networks. Recall from Section \ref{subsec:special} that in the asymptotic case $\lambda \rightarrow \infty$, the performance of an FD-node converges to the performance of the perfect LI cancellation case. From the figure, this occurs at $\lambda \approx 0.5$ for the three-node topology whereas the two-node topology requires $\lambda > 1$. Again, the difference lies in the passive suppression ability of the BSs in the three-node topology.

\begin{figure*}[t]
  \begin{minipage}{0.32\linewidth}\centering
    \includegraphics[width=\linewidth]{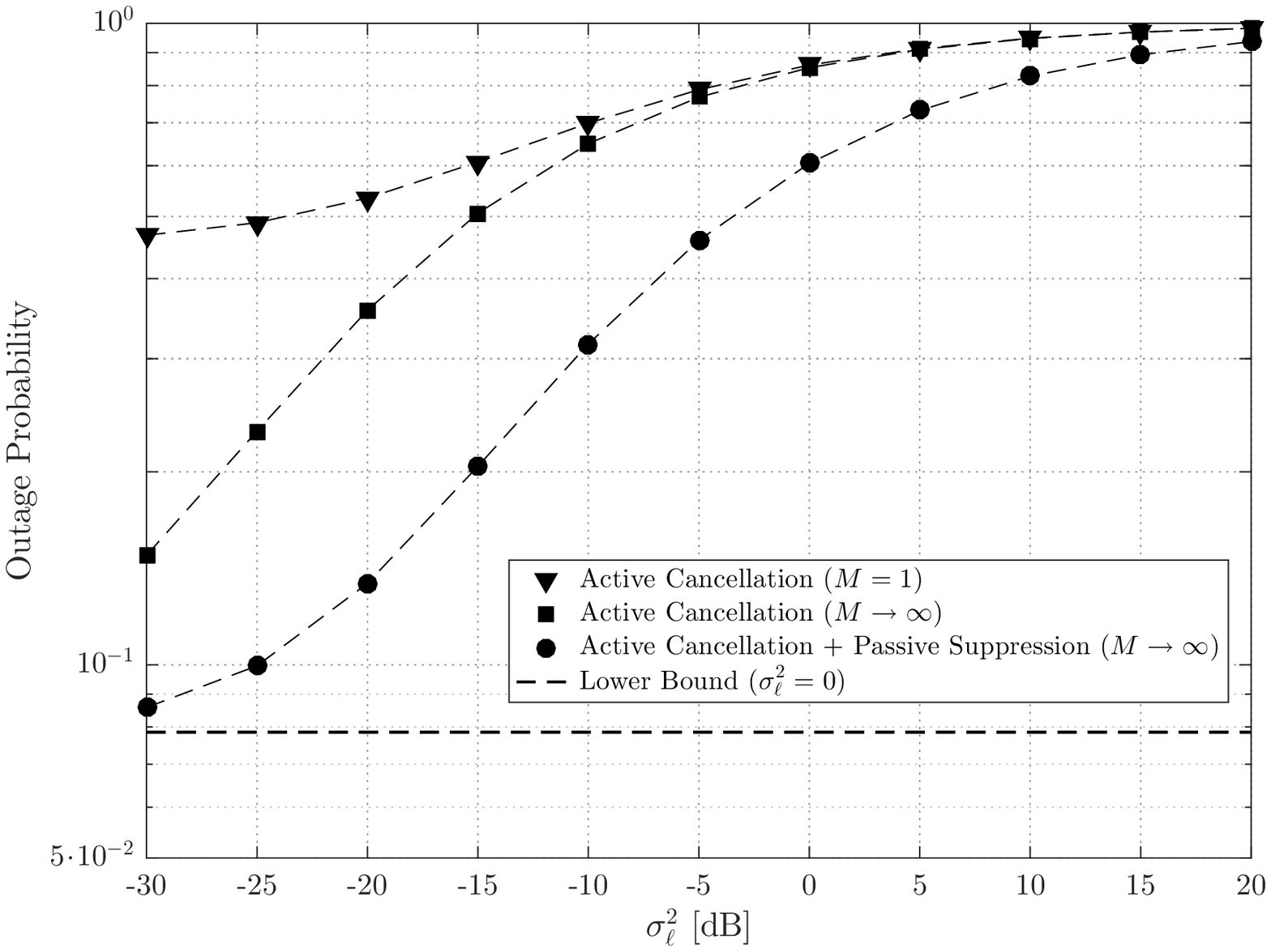}
    \captionof{figure}{Outage probability vs residual interference $\sigma_\ell^2$; $R = 0.1$ bpcu.}
    \label{fig:outage_vs_sigma}
  \end{minipage}\hspace{2mm}
  \begin{minipage}{0.32\linewidth}\centering
    \includegraphics[width=\linewidth]{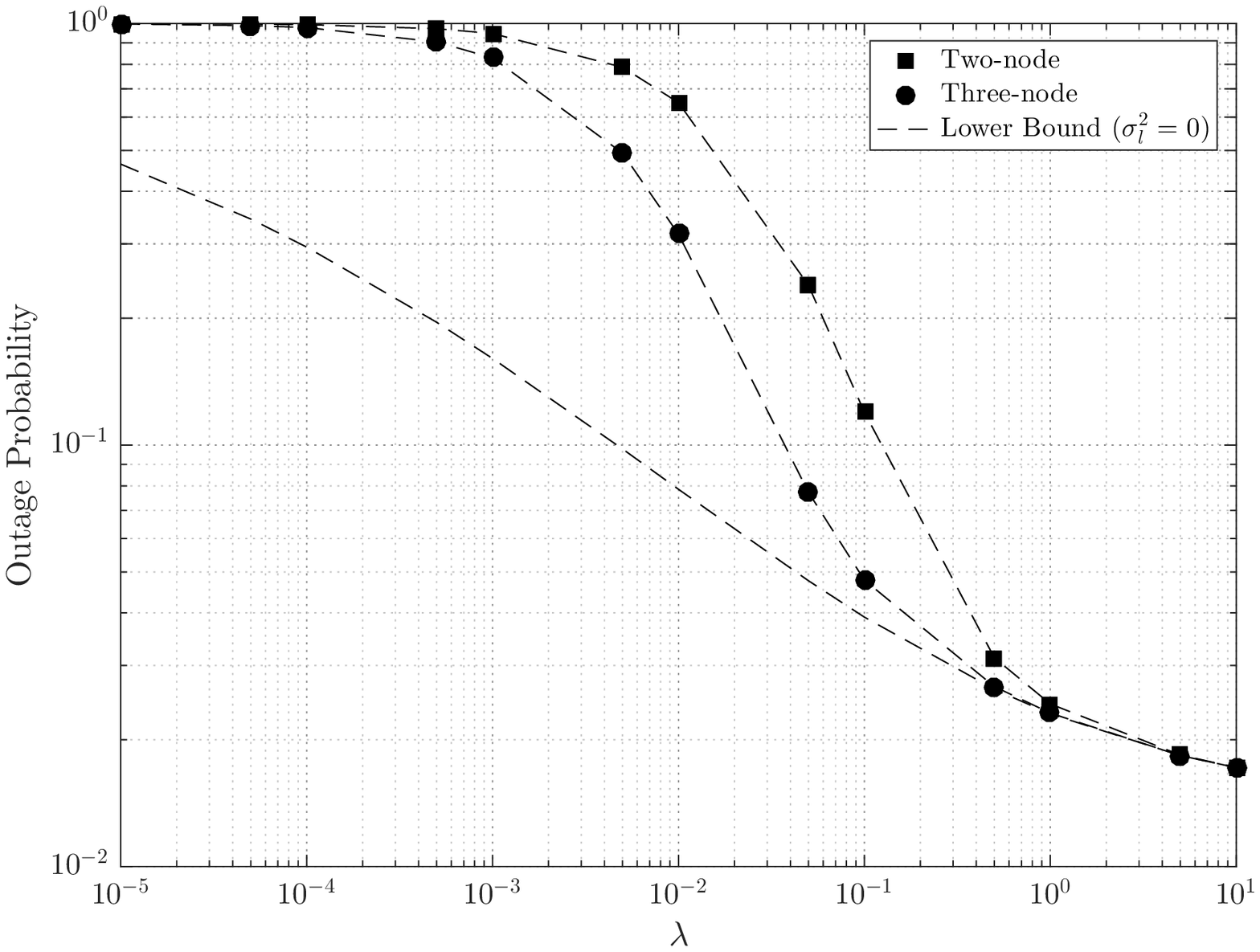}
    \captionof{figure}{Outage probability vs density $\lambda$; $R = 0.1$ bpcu, $\sigma_\ell^2 = -10$ dB.}
    \label{fig:outage_vs_density}
  \end{minipage}\hspace{2mm}
  \begin{minipage}{0.32\linewidth}\centering
    \includegraphics[width=\linewidth]{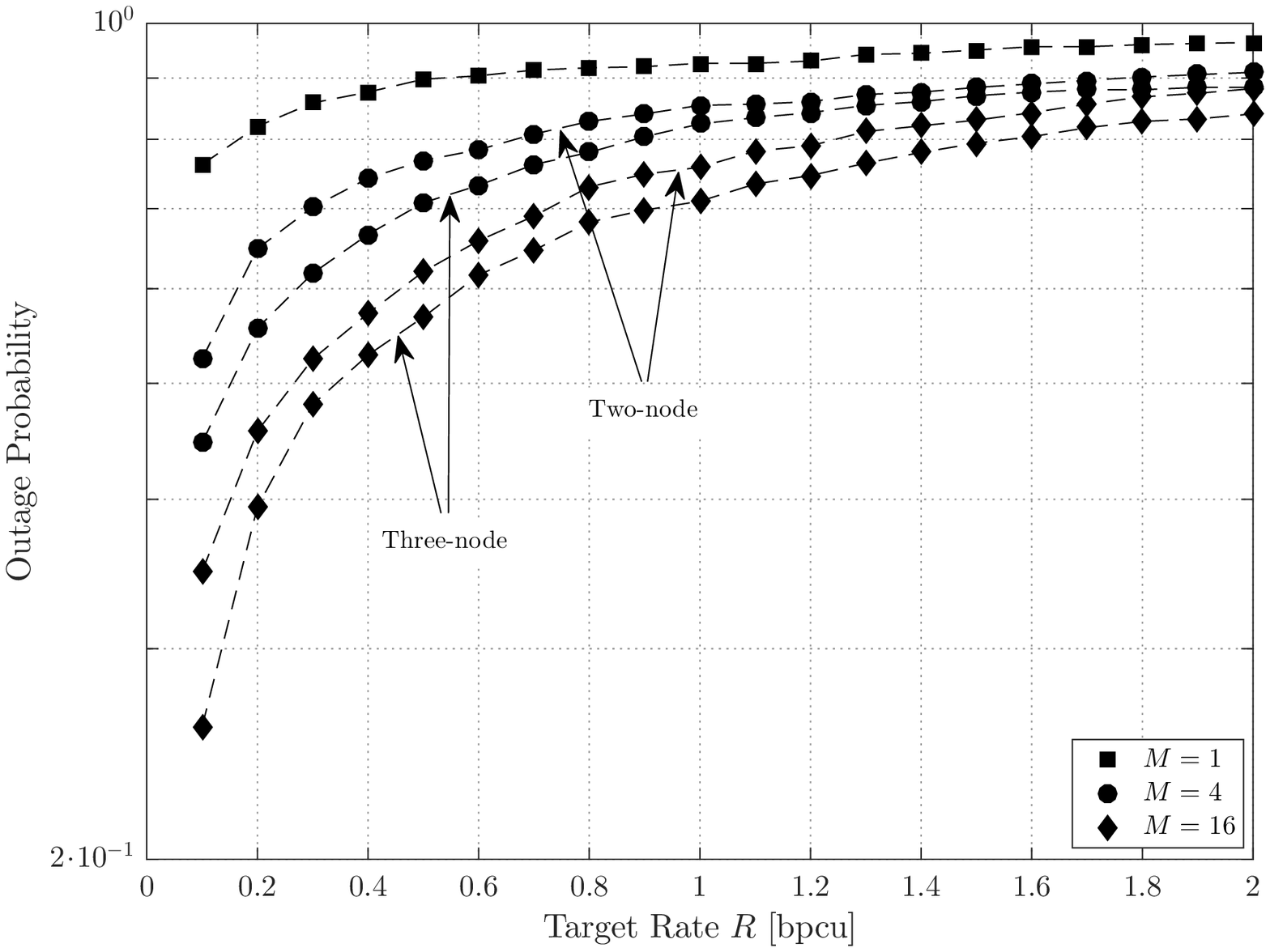}
    \captionof{figure}{Outage probability vs target rate $R$ using 3GPP model.}
    \label{fig:3gpp}
  \end{minipage}
\end{figure*}

Finally, we adopt the 3GPP model for multi-cell pico cellular networks \cite[Table 6.3-1]{3GPP} which was also used in \cite{TD} to validate our setup. In this model, the probability of a link being in line-of-sight (LOS) is given by
$p_{\rm LOS} = 0.5 - \min\left\{0.5,5\exp\left(-\frac{0.156}{r}\right)\right\}+\min\left\{0.5, 5\exp\left(-\frac{r}{0.03}\right)\right\}$, where the distance $r$ is in kilometers. The path loss $l_{\rm LOS}(r)$ between two pico base stations for the LOS case is given by a dual slope model,
\begin{equation*}
  l_{\rm LOS}(r) =
  \begin{cases}
    10^{-9.84} r^{-2}, & r < \frac{2}{3},\\
    10^{-10.19} r^{-4}, & r > \frac{2}{3},
  \end{cases}
\end{equation*}
and for the non-LOS (NLOS) case $l_{\rm NLOS}(r) = 10^{-16.94}r^{-4}$. Moreover, the path loss models between a base station and a user are $l_{\rm LOS}(r) = 10^{-10.38}r^{-2.09}$, and $l_{\rm NLOS}(r) = 10^{-14.54}r^{-3.75}$. The following parameters have been used: $P_b = 24$ dBm, $P_u = 23$ dBm, $\sigma^2_n = 5$ dB, $\sigma_\ell^2 = -30$ dB and $\lambda = 0.1$. Finally, the shadowing between two base stations is considered as lognormal with standard deviation $6$ dB and between a base station and a user has standard deviation $3$ dB for LOS and $4$ dB for NLOS. Fig. \ref{fig:3gpp} depicts the results of this model for our considered setup using simulations. It is clear from the figure that the same observations provided above can be deduced.

\section{Conclusion}\label{sec:conc}
In this paper, we have presented the impact of directional antennas on the interference mitigation in FD cellular networks. Despite the fact that the two-node architecture has by definition greater potentials, since both the BS and the user operate in FD mode, good performance is difficult to be achieved due to the inability of the terminals to restrict the residual LI. On the other hand, the three-node architecture looks more promising since, in this case, the FD-mode BS can passively suppress the LI and the HD-mode user is not affected by it. Indeed, this is also clear in the composite network where we showed that it is more beneficial for the three-node architecture to be employed in most, if not all, of the cells of the network unless the FD-mode users choose to use the uplink for small time periods. Based on our observations, we believe that the two main characteristics of the three-node architecture, i.e. passive LI suppression at the BSs and HD-mode users, makes it the most preferable architecture of the two. The three-node architecture is already regarded as the topology to be implemented first in the case of FD cellular networks due to the high complexity and energy requirements which FD will impose on future devices \cite{JSAC, GOY3}. The results of this paper are in line with this view and provide insight as to how such an architecture will perform in cellular networks with FD capabilities.

\appendices
\section{Proof of Theorem \ref{thm:out_prob_2d}}\label{prf:out_prob_2d}
To derive the outage probability, i.e. the cumulative distribution function of the SINR, we take the expectation over both small- and large- scale fading \cite{HAEN1, HAEN2}; decoupling the two is left as a potential future direction \cite{GG}. Therefore, conditioning on the nearest BS being at a distance $r$ we have,
\begin{align*}
P_{\TWD} &= \E_r \left[\PP\left[\log_2\left(1+\sinr\right) < R ~ | ~ r\right]\right] = \int_0^\infty \PP[\log_2(1+\sinr) < R ~ | ~ r]f_r(r)\dd r\\
&=1-2\pi \la\int_0^\infty re^{-\la \pi r^2}\PP[\sinr \geq 2^R-1 ~ | ~ r]\dd r.
\end{align*}
Letting $\tau = 2^R-1$, $\PP[\sinr \geq \tau ~ | ~ r]$ is the coverage probability conditioned on the distance $r$ and is given by,
\begin{align}\nonumber
\PP[\sinr \geq \tau ~ | ~ r] &=
\PP\left[|h|^2 \geq \frac{\tau r^{\alpha_1}}{P_b\G_{u,b,1}}(\sigma_n^2 + \mathds{1}_{\rm FD}I_\ell + I_b + I_u) ~ \Bigg| ~ r\right]\\\label{eq:exp}
&\stackrel{(a)}{=} \E_{I_\ell,\Phi,\Psi}\left[\exp\left(-\frac{ \tau r^{\alpha_1}}{P_b\G_{u,b,1}}(\sigma_n^2 + I_\ell + I_b + I_u)\right)\right]\\\nonumber
&\stackrel{(b)}{=} \exp\left(-\frac{\tau r^{\alpha_1}}{P_b\G_{u,b,1}}\sigma_n^2\right) \E_{I_\ell}\left[e^{-\frac{\tau r^{\alpha_1}}{P_b\G_{u,b,1}}I_\ell}\right]
\E_{I_b}\left[e^{-\frac{\tau r^{\alpha_1}}{P_b\G_{u,b,1}}I_b}\right]
\E_{I_u}\left[e^{-\frac{\tau r^{\alpha_1}}{P_b\G_{u,b,1}}I_u}\right]\\\nonumber
&\stackrel{(c)}{=} \exp\left(-\frac{\tau r^{\alpha_1}}{P_b\G_{u,b,1}}\sigma_n^2\right) \La^{\TWD}_{I_\ell}\left(\frac{\tau r^{\alpha_1}}{P_b\G_{u,b,1}}\right)
\La^{\TWD}_{I_b}\left(\frac{\tau r^{\alpha_1}}{P_b\G_{u,b,1}}\right)
\La^{\TWD}_{I_u}\left(\frac{\tau r^{\alpha_1}}{P_b\G_{u,b,1}}\right),
\end{align}
where $(a)$ follows from the fact that $|h|^2 \sim \exp(1)$ and $\mathds{1}_{\rm FD} = 1$ since the receiver in the two-node architecture is FD-capable; $(b)$ follows from the independence between $\Phi_b$ and $\Phi_u$ (and therefore the independence between $I_b$ and $I_u$); $(c)$ $\La_I(s)$ denotes the Laplace transform of the random variable $I$ evaluated at $s$. Using $s = \frac{\tau r^{\alpha_1}}{P_b\G_{u,b,1}}$, the Laplace transform of $I_\ell$ can be derived from the moment generating function (MGF) of an exponential variable since $I_\ell = P_u \G_{u,u,1} |h_\ell|^2$ and $|h_\ell|^2 \sim \exp(1/\sigma^2_\ell)$. Therefore,
\[\La^{\TWD}_{I_\ell}\left(s\right) = \frac{1}{1+ s P_u \G_{u,u,1} \sigma_\ell^2}.\]
As there is no intra-cell interference, $\La_{I_u}(s)$ needs to be evaluated conditioned on the distance, say $\rho$, from $u_0$ to the closest uplink user in the neighbouring cells. Since the densities of $\Phi$ and $\Psi$ are equal and $\rho$ is independent of $r$, we assume $\rho$ is distributed according to the pdf of the distance to the nearest neighbour as given in Section \ref{sec:model} \cite{CP, HT}. Thus the Laplace transform of $I_u$ is given by,
\begin{align}\label{eq:laplace_I_u}
\La^{\TWD}_{I_u}(s) = \E_{I_u}[e^{-sI_u} ~ | ~ \rho] = \int_0^\infty \E_{I_u}[e^{-sI_u}]f_{\rho}(\rho)\dd \rho.
\end{align}
The expected value is then evaluated as follows,
\begin{align}
\nonumber \E_{I_u}[e^{-sI_u}] &=
\prod_{i\in\{1,2,3,4\}} \E_{\Psi_i, |k_j|^2}\left[\exp(-sP_u \G_{u,u,i} \sum_{j\in\Psi_i} |k_j|^2 D_j^{-\alpha_2})\right]\\
&\stackrel{(a)}{=} \prod_{i\in\{1,2,3,4\}} \E_{\Psi_i}\left[\prod_{j\in\Psi_i}\E_k\left[\exp\left(-sP_u \G_{u,u,i} |k|^2 D_j^{-\alpha_2}\right)\right]\right]\nonumber\\\nonumber
&\stackrel{(b)}{=} \prod_{i\in\{1,2,3,4\}}\exp\left\{-2\pi\la_{u,u,i} \int_\rho^\infty \left(1-\E_{k}\left[\exp\left(-sP_u \G_{u,u,i} |k|^2 y^{-\alpha_2}\right)\right] \right) y \dd y \right\}\\
&\stackrel{(c)}{=} \prod_{i\in\{1,2,3,4\}}\exp\left\{-2\pi\la_{u,u,i} \int_\rho^\infty \left(1 - \frac{1}{1 + s P_u \G_{u,u,i} y^{-\alpha_2}} \right) y \dd y \right\}\label{eq:int_approx}\\
&\stackrel{(d)}{=} \prod_{i\in\{1,2,3,4\}} \exp\left\{-\frac{2\pi \la_{u,u,i} }{\alpha_2-2} \G_{u,u,i} F\left(\alpha_2, \frac{s \G_{u,u,i} P_u}{ \rho^{\alpha_2}}\right)s P_u \rho^{2-\alpha_2}\right\},\label{eq:exp_I_u}
\end{align}
where $(a)$ follows from the fact that $|k_j|^2$ are independent and identically distributed and also independent from the point process $\Psi$; $(b)$ follows from the probability generating functional (PGFL) of a PPP \cite{HAEN1} and the limits are from $\rho$ to $\infty$ since the closest interfering uplink user is at least at a distance $\rho$; $(c)$ follows from the MGF of an exponential random variable and since $|k|^2 \sim \exp(1)$ and finally $(d)$ is based on \cite[Eq. (3.194.2)]{GRA}.

\noindent Replacing $\E_{I_u}[e^{-sI_u}]$ with \eqref{eq:exp_I_u} gives,
\begin{align*}
\La^{\TWD}_{I_u}\left(s\right) &= 2\pi \la\int_0^\infty \rho e^{-\la \pi \rho^2}
\prod_{i\in\{1,2,3,4\}} \exp\left\{-\frac{2\pi \la_{u,u,i}}{\alpha_2-2} \G_{u,u,i} F\left(\alpha_2, \frac{s \G_{u,u,i} P_u}{ \rho^{\alpha_2}}\right)s P_u \rho^{2-\alpha_2}\right\}\dd \rho\\
&=2\pi \la\int_0^\infty \rho \exp\left\{(-\pi \rho^2\left(\la + \sum_{i\in\{1,2,3,4\}} \frac{2\la_{u,u,i}}{\alpha_2-2} \G_{u,u,i} F\left(\alpha_2, \frac{s \G_{u,u,i} P_u}{\rho^{\alpha_2}}\right)\frac{sP_u}{\rho^{\alpha_2}}\right)\right\}\dd \rho.
\end{align*}
The Laplace transform of $I_b$ can be derived similarly to above where the limits of the integral are from $r$ to $\infty$ since the nearest interfering BS is further from the associated BS. This gives,
\[\La^\TWD_{I_b}\left(s\right) = \prod_{i\in\{1,2,3,4\}}\exp\left\{-\frac{2\pi\la_{u,b,i}}{\alpha_1-2} \frac{\G_{u,b,i}}{\G_{u,b,1}} F\left(\alpha_1, \frac{\G_{u,b,i}}{\G_{u,b,1}} \tau\right)r^2 \tau \right\},\]
and the result follows.\vspace{-3mm}

\section{Passive Suppression Function}\label{pass_suppress_desc}

The passive suppression function $\supfn$ has been derived based on the experimental results in \cite{EE1}. It provides the level of passive suppression that can be achieved at a certain angle $\theta$ between the transmit and receive antenna; $\theta_{\rm max}$ is where the maximum suppression occurs. The smaller the value of $\supfn$ the better, so $\supfn = 0$ refers to perfect passive suppression and $\supfn = 1$ refers to no passive suppression, i.e., when the transmit and receive antenna operate in the same sector ($\theta = 0$), which is always true for the two-node architecture. The function never actually takes the value of $0$ as passive suppression mitigates the loopback interference but cannot erase it completely. The cosine difference $\cos\left(\theta_{\rm max}\right)-\cos\left(|\theta|-\theta_{\rm max}\right)$ was chosen due to the symmetry obtained around $\theta = 0$ and since it provides the lowest value at $\theta_{\rm max}$. Note that there are most likely many other functions that can provide a similar behaviour. As $\theta$ gets smaller the passive suppression ability diminishes since the coupling between the two antennas becomes stronger \cite{EE1}. This behaviour is captured by the cosine difference. The exponential function was chosen in order to obtain a value of $1$ at $\theta = 0$ and since it always provides a positive value. Essentially, any exponential function would produce a similar behaviour but the natural exponential function was chosen due to its popularity. We assume that $\theta_{\rm max}$ increases with the antenna efficiency and so for small $\theta_{\max}$ the achievable passive suppression is generally low and for most angles there is no passive suppression. In these cases, $\supfn$ may take values greater than $1$ and thus the $\min$ operator was used.\vspace{-2mm}

\section{Proof of Theorem \ref{thm:out_prob_3u}}\label{prf:out_prob_3u}
The proof of Theorem \ref{thm:out_prob_3u} follows similar steps as the proofs of Theorems \ref{thm:out_prob_2d} and \ref{thm:out_prob_3d}. The Laplace transform of $I_b$ is derived in the same way as \eqref{eq:lap_int_u_2d}. Also, the Laplace transform of $I_u$ is evaluated similarly to \eqref{eq:lap_int_u_2d} but the interference fields $\Psi_i$ in this case are inhomogeneous PPPs with density function $\la_{b,u,i} \left(1-\exp(-\pi\la x^2)\right)$ \cite{GOY2, JA2} which ensures that $\Psi_i$ contain only the uplink users from other cells. Therefore, using the same steps as in Appendix \ref{prf:out_prob_2d} we have
\begin{align*}
&\E_{I_u}[e^{-sI_u}] = \prod_{i\in\{1,2,3,4\}} \E_{\Psi_i}\left[\prod_{j\in\Psi_i}\E_k\left[\exp\left(-sP_u \G_{b,u,i} |k|^2 D_j^{-\alpha_2}\right)\right]\right]\\
&= \prod_{i\in\{1,2,3,4\}}\exp\left\{-2\pi\la_{b,u,i} \int_0^\infty \left(1-\exp(-\pi\la y^2)\right)\left(1-\E_{k}\left[\exp\left(-sP_u \G_{b,u,i} |k|^2 y^{-\alpha_2}\right)\right] \right) y \dd y \right\}\\
&= \prod_{i\in\{1,2,3,4\}}\exp\left\{-2\pi\la_{b,u,i} \left(\int_0^\infty \frac{s P_u \G_{b,u,i}}{s P_u \G_{b,u,i} + y^{\alpha_2}} y\dd y + \int_0^\infty \frac{s P_u \G_{b,u,i}}{s P_u \G_{b,u,i} + y^{\alpha_2}} e^{-\pi\la y^2} y \dd y\right) \right\},
\end{align*}
where $s = \frac{\tau r^{\alpha_1}}{P_u G_b G_u}$; the first integral is evaluated similarly to the proof of expression \eqref{eq:lap_int_u_3d} and the result follows by the change of variable $y^2 = z$. Finally, for $\La_{I_\ell}(s) = \E\left[e^{-sI_\ell}\right]$ where $I_\ell$ is given by \eqref{eq:li_3u}, we have
\begin{align}
\La^{\THU}_{I_\ell}\left(s\right)
&= \E\left[\exp\left(-s P_u G_b^2 |h_\ell|^2 \supfn(B_0 + \g_b (1-B_0))\right)\right]\\
&= \frac{1}{M_b}\E\left[\exp\left(-s P_u G_b^2 |h_\ell|^2\right)\right] + \frac{M_b-1}{M_b}\E\left[\exp\left(-s P_u G_b H_b |h_\ell|^2 \supfn\right)\right],\label{eq:bern_step}
\end{align}
with $\theta \neq 0$. Expression \eqref{eq:bern_step} follows from $f_\ell(0, \theta_{\rm max}) = 1$, $\theta_{\rm max} \in \left[\frac{\pi}{2},\pi\right]$  and the Bernoulli random variable $B_0$ with parameter $\displaystyle\frac{1}{M_b}$. Since each angle $\theta \neq 0$ between the two sectors occurs with probability $\displaystyle\frac{1}{M_b}$ and using the MGF of an exponential random variable we have,
\[\La^{\THU}_{I_\ell}(s) = \frac{1}{M_b}\frac{1}{1+ s P_b G_b^2 \sigma_\ell^2} +
\frac{1}{M_b}\mathlarger{\sum}\limits_{\substack{\theta \in [-\pi, \pi) \setminus \{0\}\\\theta \equiv 0 \Mod{\frac{2\pi}{M_b}}}}\frac{1}{1 + s P_b G_b H_b \sigma_\ell^2 \supfn}.\]
By substituting \eqref{eq:ang_fun} to the above expression we get \eqref{eq:lap_li_3u}.

\section{Proof of Proposition \ref{prop:out_multi_down}}\label{prf:out_multi_down}
Starting from \eqref{eq:exp} in Appendix \ref{prf:out_prob_2d} with $\sigma_n^2 = 0$ we have,
\begin{align*}
&\PP[\sinr \geq \tau ~ | ~ r]
= \E_{I_\ell,\Phi,\Psi}\left[\exp\left(-\frac{\tau r^{\alpha_1}}{P_b\G_{u,b,1}}(I_\ell + I_b + I_u)\right)\right]\\
&= p^{\TWN}\E_{I_\ell,\Phi,\Psi^{\TWN}}\left[\exp\left(-\frac{\tau r^{\alpha_1}}{P_b\G_{u,b,1}}(I_\ell + I_b + I_u)\right)\right]+ p^{\THN}\E_{\Phi,\Psi^{\THN}}\left[\exp\left(-\frac{\tau r^{\alpha_1}}{P_b\G_{u,b,1}}(I_b + I_u)\right)\right]\\
&= p^{\TWN}\La^{\TWD}_{I_b}\left(\frac{\tau r^{\alpha}}{P_b\G_{u,b,1}}\right)\La^{\TWD}_{I_\ell}\left(\frac{\tau r^{\alpha}}{P_b\G_{u,b,1}}\right) \La^{\TWD}_{I_u}\left(\frac{\tau r^{\alpha}}{P_b\G_{u,b,1}}\right)
 + p^{\THN} \La^{\THD}_{I_b}\left(\frac{\tau r^{\alpha}}{P_b\G_{u,b,1}}\right) \La^{\THD}_{I_u}\left(\frac{\tau r^{\alpha}}{P_b\G_{u,b,1}}\right),
\end{align*}
where $\La_{I_b}\left(\frac{\tau r^{\alpha}}{P_b\G_{u,b,1}}\right)$, $\La^{\TWD}_{I_u}\left(\frac{\tau r^{\alpha}}{P_b\G_{u,b,1}}\right)$ and $\La^{\THD}_{I_u}\left(\frac{\tau r^{\alpha}}{P_b\G_{u,b,1}}\right)$ are given by \eqref{eq:lap_int_b_2d}, \eqref{eq:lap_int_u_2d_multi} and \eqref{eq:lap_int_u_3d_multi} respectively. Since $\Pi_d = 1 - 2\pi \la\int_0^\infty \PP[\sinr \geq \tau ~ | ~ r] re^{-\la \pi r^2}\dd r$ the result follows.

\end{document}